\documentclass[11pt,letterpaper]{article}

\usepackage[margin=0.65in]{geometry}
\usepackage{longtable}
\usepackage{xltabular}
\usepackage{threeparttable}
\usepackage{multirow}
\usepackage{booktabs}
\usepackage{graphicx}
\usepackage{amsmath,amssymb,amsfonts}
\usepackage{natbib}
\usepackage{setspace}
\usepackage{authblk}
\usepackage{tcolorbox}
\usepackage[colorlinks=true, bookmarksopen=true, bookmarksnumbered=true,
            citecolor=blue, linkcolor=blue, urlcolor=blue]{hyperref}
\usepackage[capitalize,noabbrev]{cleveref}
\usepackage{amsthm}

\newtheorem{theorem}{Theorem}[section]
\newtheorem{lemma}[theorem]{Lemma}
\newtheorem{corollary}[theorem]{Corollary}
\newtheorem{proposition}[theorem]{Proposition}
\theoremstyle{definition}

\theoremstyle{remark}

\newtheorem*{remark*}{Remark}

\DeclareMathOperator{\Var}{Var}
\DeclareMathOperator{\Cov}{Cov}
\DeclareMathOperator{\Corr}{Corr}
\DeclareMathOperator{\diag}{diag}
\DeclareMathOperator{\Mult}{Mult}
\DeclareMathOperator{\Bias}{Bias}
\DeclareMathOperator{\SE}{SE}

\hypersetup{pdfencoding=auto}
\begin{document}

\title{A Utility Score Framework for Dose Optimization Studies with Binary Efficacy-Safety Endpoints: Sample Size Determination and Bias Characterization}

\author[1]{Xuemin Gu\thanks{Corresponding author. Email: xgu@relaytx.com
This work was conducted while the author was at Bicycle Therapeutics.}}
\author[1]{Cong Xu}
\author[2]{Lei Xu}
\author[3]{Ying Yuan}
\affil[1]{Relay Therapeutics, 60 Hampshire St, Cambridge, MA 02139, USA}
\affil[2]{Bicycle Therapeutics, 4 Hartwell Place, Lexington, MA 02421, USA}
\affil[3]{Department of Biostatistics, The University of Texas MD Anderson Cancer Center, Houston, TX, USA}
\date{\today}
\maketitle

\begin{abstract}
The FDA's \textit{Project Optimus} initiative emphasizes patient-centered dose selection in oncology that balances efficacy and safety. We develop a design framework for randomized dose optimization studies that uses clinically interpretable utility scores to integrate binary efficacy and safety endpoints. From two candidate doses identified in prior studies, our design uses the observed mean utility scores to select the superior dose. A systematic approach for eliciting clinically meaningful utility scores from efficacy-safety trade-offs is proposed to ensure that design parameters reflect actual clinical priorities. We provide closed-form sample size formulas to achieve prespecified Probabilities of Correct Selection (PCS) under clinically relevant scenarios. The framework generalizes the one stage option of Randomized Optimal Selection (ROSE) design as a special case where utility depends solely on efficacy. We also derive analytical expressions for selection-induced bias, which characterizes how bias from the dose optimization stage may propagate to subsequent confirmatory trials. This bias analysis enables prospective planning for Type~I error control in follow-on stage 2 confirmatory studies.
Extensive simulations ($10^6$ replications per scenario) across diverse scenarios confirm that the sample size formulas reliably achieve target PCS and that analytical bias and Type I error formulas closely approximate empirical estimates. An R package \texttt{DoseOptDesign} implementing the proposed methods is publicly available.
\end{abstract}


\section{Introduction}
\label{sec:introduction}

The conventional paradigm of oncology drug development was largely built around cytotoxic chemotherapy agents. With these agents, increasing dose levels may enhance treatment efficacy, but they often come with a corresponding decline in safety. Consequently, the primary objective of Phase I first-in-human trials has been to identify the maximum tolerated dose (MTD), defined as the highest dose that achieves clinically significant efficacy while maintaining acceptable toxicity. In contrast, many recent cancer therapies---such as targeted therapies, immunotherapies, and other personalized medicines---demonstrate clinically meaningful effects at doses lower than the MTD \citep{sachs2016optimal, yan2018optimal}. This allows for extended administration of these agents, making their improved tolerability at lower doses particularly optimal for long-term efficacy.

To facilitate the development of these innovative treatments, the FDA's Oncology Center for Excellence has launched the \textit{Project Optimus} initiative \citep{shah2021optimus}, which emphasizes a patient-centered approach to dose optimization. The objective is no longer to find the highest tolerable dose, but to identify an optimal dose that balances efficacy and safety. The FDA's guidance encourages randomized dose-finding trials sized not for formal hypothesis testing, but to ensure precision in dose selection based on predefined clinical criteria \citep{FDA2024}.

A challenge in designing such trials is how to formally integrate efficacy and safety into a single decision framework. Utility scores, which assign weighted values to combined outcomes, provide an established method for this purpose. Methodological development in this area began with phase I/II designs like EffTox that balanced efficacy-toxicity trade-offs \citep{thall2004efftox}. Subsequent innovations embedded utility scores within Bayesian frameworks (e.g., U-BOIN, BOIN12) to guide dose escalation for targeted therapies \citep{zhou2019uboin,lin2020boin}. More recently, methods have expanded to incorporate adaptive strategies, multi-endpoint settings, and patient-reported outcomes \citep{jiang2023seamless,yuan2024statistical,dangelo2024umet,alger2024upro}.

Despite these advances, two statistical questions remain underexplored, particularly for the common setting of randomized trials with binary endpoints. The first is the lack of a straightforward Frequentist methodology for sample size determination to achieve a desired Probability of Correct Selection (PCS) using utility score-based methods, which co-optimize both efficacy and safety. The second concerns the consequences of the utility score-based dose selection process itself, specifically the selection-induced estimation bias in the efficacy endpoint that propagates to confirmatory trials.

Regarding the first underexplored question for dose optimization studies, the Randomized Optimal Selection (ROSE) design \citep{wang2025rose} offers a complementary approach for sample size calculation, adapting the selection design of Simon et al.\ \citep{simon1985randomized} for ordered doses. ROSE minimizes sample size while ensuring a target PCS for the optimal biological dose (OBD). However, it handles safety only implicitly---favoring lower doses unless efficacy gains justify higher ones---rather than through an explicit utility function. This reliance on an indirect safety mechanism means that while ROSE addresses the sample size problem for efficacy-based selection, it leaves open the corresponding challenges of sample size calculation and bias estimation for designs based on an explicit utility function.

While the statistical properties of selection-induced bias have been characterized for efficacy-based treatment selection \citep{cohen1989two, stallard2008estimation, bowden2008unbiased, bauer2010selection}, these prior works did not address the distinct problem where selection is based on a composite utility score, while the response rate remains the endpoint of interest for future confirmatory trials. The estimation bias of such a mechanism—which is driven by the covariance between the utility score and the final efficacy endpoint—has not been previously characterized. This is the case whether the confirmatory endpoint is a component of the utility function (e.g., response) or an externally correlated endpoint (e.g., progression-free survival).

This paper focuses on a one-stage, two-sample randomized dose optimization study, a design relevant for the common scenario in which a study is conducted immediately before a large-scale Phase 3 registration trial. In this late-stage setting, a one-stage approach is often preferred for two reasons. First, sufficient clinical data are available to narrow the choice to two strong candidate doses, making a single definitive comparison the most direct path to a dose optimization decision. Second, a one-stage design avoids the operational complexity and potential delays of multi-stage approaches, ensuring a timely transition to the pivotal registration trial. 

To support this practical setting, we make three primary contributions: a systematic method for specifying utility scores that directly quantify clinical trade-offs; a methodology for sample size determination to achieve prespecified PCS; and analytical expressions for selection-induced bias. We also demonstrate that the one-stage option of ROSE design is a special case of our utility framework. These three contributions form a statistical framework for designing dose optimization studies that align with the principles of Project Optimus and for quantifying the impact of dose selection on subsequent confirmatory trials.

The remainder of the paper is organized as follows. Section~\ref{sec:framework} presents the utility score framework and dose selection design. Section~\ref{sec:sample_size} describes sample size determination. Section~\ref{sec:bias} provides a detailed characterization of selection-induced bias, with special emphasis on its propagation to time-to-event confirmatory endpoints. Section~\ref{sec:simulations} presents extensive simulation studies, and Section~\ref{sec:discussion} concludes with practical implications and future directions.

\section{Utility Score Framework and Dose Selection Design}
\label{sec:framework}

\subsection{Overview and Motivation}
\label{subsec:design}

We propose a randomized dose optimization design to identify a single clinically optimal dose from two candidate dose levels or dosing regimens. This design is motivated by the FDA's \textit{Project Optimus} initiative \citep{shah2021optimus, FDA2024}, which advocates for patient-centered dosing through balanced assessment of efficacy and safety. The design is intended for settings where early-phase studies have identified two candidate doses with acceptable safety profiles, and a more rigorous evaluation is needed to select the one with the most favorable benefit-risk profile for subsequent trials.

Unlike dose selection strategies rooted in the traditional MTD-seeking paradigm of dose escalation studies, this framework directly adopts the principle of \textit{Project Optimus}, which states that the optimal doses for non-cytotoxic agents (e.g., targeted therapies, immunotherapies) must balance efficacy and tolerability, rather than simply maximize exposure. Our proposed method facilitates the explicit quantification of efficacy-safety trade-offs.

Patients are randomized to two candidate doses: Dose $L$ (Low) and Dose $H$ (High). To illustrate our methodology, we use the following motivating example throughout this paper:
\begin{itemize}
    \item Dose $L$: 5 mg/m$^2$ weekly (QW)
    \item Dose $H$: 8 mg/m$^2$ every two weeks (Q2W)
\end{itemize}
Here, Dose $L$ delivers a higher cumulative dose over a 4-Week cycle, while Dose $H$ results in a higher peak concentration ($C_{\max}$). So, `Low' and `High' are labels for the nominal dose levels to distinguish the two regimens, not a presumption about their relative safety or efficacy.

\subsection{Efficacy-Safety Outcomes and Utility Framework}

Each treated patient provides two binary outcomes:
\begin{itemize}
    \item \textbf{Efficacy endpoint ($X$):} Binary indicator of clinical response
    \begin{equation}
    X = \begin{cases} 1, & \text{Response achieved (e.g., radiographic response, biomarker threshold)} \\ 0, & \text{No response} \end{cases}
    \label{eq:efficacy_def}
    \end{equation}
    
    \item \textbf{Safety endpoint ($Y$):} Binary indicator of absence of certain adverse events of interest
    \begin{equation}
    Y = \begin{cases} 1, & \text{No adverse event} \\ 0, & \text{Adverse event occurred} \end{cases}
    \label{eq:safety_def}
    \end{equation}
\end{itemize}
We denote the response rate by $p = \Pr(X=1)$ and the no-AE rate by $q = \Pr(Y=1)$. These marginal probabilities, together with the correlation $\phi$ between $X$ and $Y$, fully determine the joint distribution of efficacy-safety outcomes.

This binary formulation accommodates categorical endpoints and continuous measurements dichotomized at prespecified threshold values. The convention that $Y=1$ denotes the absence of adverse events, along with the definition of efficacy ($X=1$ for response), ensures that higher values for both endpoints are clinically desirable.

The traditional assumption that higher exposure improves efficacy is often unreliable for novel anticancer therapies. Furthermore, the relationship between exposure and safety is complex, as adverse events may be driven by either peak concentration ($C_{\max}$) or cumulative exposure. This complexity necessitates a framework that directly evaluates the efficacy-safety trade-off without assuming a monotonic dose-response relationship. 

\begin{table}[ht]
\centering
\caption{Indicator Vector, $\mathbf{W}$, for Efficacy-Safety Outcome Combinations and Assigned Utility Scores}
\label{tab:utility01}
\resizebox{\textwidth}{!}{
\begin{tabular}{lccccc}
\toprule
Outcome & Efficacy ($X$) & Safety ($Y$) & Interpretation & $\mathbf{W}=(W_1,W_2,W_3,W_4)$ & Utility Score ($u_k$) \\
\midrule
1 (Best) & 1 & 1 & Response, No AE & (1, 0, 0, 0) & $u_1$\\
2 & 1 & 0 & Response, AE Present & (0, 1, 0, 0) & $u_2$\\
3 & 0 & 1 & No Response, No AE & (0, 0, 1, 0) & $u_3$\\
4 (Worst) & 0 & 0 & No Response, AE Present & (0, 0, 0, 1) & $u_4$\\
\bottomrule
\end{tabular}
}
\end{table}

To address this, our design employs a Utility Score approach. As shown in Table \ref{tab:utility01}, an indicator random vector $\mathbf{W}$ and a utility score are utilized to model and quantify the four mutually exclusive outcomes of the binary efficacy ($X$) and safety ($Y$) endpoints, respectively. The random vector $\mathbf{W} = (W_1, W_2, W_3, W_4)$ follows a Categorical distribution with probability parameter $\mathbf{\pi} = (\pi_1, \pi_2, \pi_3, \pi_4)$, where $\sum_{k=1}^4 \pi_k = 1$. This is equivalent to a Multinomial distribution with a size parameter of 1, and $\pi_k$ represents the probability that $W_k=1$, $k \in \{1,2,3,4\}$. The utility score $u_k$ assigned to each outcome must reflect clinical priorities:
\begin{equation}
u_1 \geq u_2 \geq u_3 \geq u_4.
\label{eq:utility_order01}
\end{equation}
This hierarchy formalizes the clinical judgment that a response with no toxicity ($u_1$) is the best, while a non-response with toxicity ($u_4$) is the worst.

At the end of the study, the \textbf{mean utility score} for each dose is computed as the average utility score across all treated patients:
\begin{equation}
\bar{U}_j = \frac{1}{n} \sum_{i=1}^{n} U_{ij},
\label{eq:observed_u}
\end{equation}
where $U_{ij}$ is the utility score for patient $i$ on dose $j$, $j \in \{L, H\}$. As a random variable, $U_{ij}$ takes one of the values $\{u_1, u_2, u_3, u_4\}$ depending on the patient's efficacy and safety outcomes.

\subsection{Utility Score Specification: The Margin-Based Approach}
\label{subsubsec:utility_specification}

The utility scores ($u_1, u_2, u_3, u_4$) are prespecified design parameters that quantify the clinical value of each of the four possible efficacy--safety outcome combinations in Table~\ref{tab:utility01}. While these scores can be determined through direct elicitation from clinicians, we recommend a \textbf{margin-based approach} that derives the utilities from two clinically meaningful quantities:
\begin{equation}
\left\{ 
\begin{aligned}
\delta > 0 &: \text{ a clinically meaningful margin of \textit{efficacy} (response)}, \\
d > 0 &: \text{ a clinically meaningful margin of \textit{safety} (no-AE)}.
\end{aligned}
\right.
\label{eq:clinical_margins}
\end{equation}

These margins are typically elicited by posing trade-off questions to clinical experts, such as: ``What improvement in response rate would justify accepting a given increase in adverse event rate?'' This approach ensures that the utility scores directly encode actual clinical priorities rather than abstract numerical values.

\paragraph{Clinical Motivation: Utility Independence.}
We adopt a simplifying assumption that the incremental value of achieving a response should not depend on whether the patient also experiences an adverse event. Mathematically, this requires:
\begin{equation}
u_1 - u_3 = u_2 - u_4,
\label{eq:utility_independence}
\end{equation}
or equivalently, $u_1 - u_2 - u_3 + u_4 = 0$. This property, known as \textit{utility independence}, is a reasonable simplification in many clinical settings, particularly in oncology where achieving a response is of primary importance regardless of toxicity status.

\paragraph{Derivation of Utility Scores.}
Without loss of generality, we anchor the utility scale by setting $u_1 = 1$ (best outcome) and $u_4 = 0$ (worst outcome). The utility independence constraint then reduces to $u_2 + u_3 = 1$, leaving one degree of freedom to specify.

The clinical trade-off ratio
\begin{equation}
r = \frac{\delta}{d}
\label{eq:tradeoff_ratio}
\end{equation}
is the ratio of the two clinically meaningful thresholds. Since $\delta$ and $d$ each represent one ``unit'' of clinical significance in their respective domains, a $\delta$ efficacy improvement and a $d$ safety improvement should contribute equally to the mean utility to keep a consistent clinical interpreation of utility score. This requires $u_2 \cdot \delta = u_3 \cdot d$, or equivalently, $u_3/u_2 = r$, where $u_2$ (``response with AE'') and $u_3$ (``no response, no AE'') represent the utility contributions due to pure efficacy and pure safety relative to the worst outcome ($u_4$), respectively. 

Solving these constraints yields:
\begin{equation}
\boxed{u_1 = 1, \quad u_2 = \frac{1}{1+r}, \quad u_3 = \frac{r}{1+r}, \quad u_4 = 0.}
\label{eq:utility_formula}
\end{equation}
The value of $r$ determines the relative importance of efficacy versus safety:
\begin{itemize}
    \item $r < 1$ ($\delta < d$): AEs are tolerable, so $u_2 > u_3$—achieving a response, even with an AE, is valued more than remaining safe without response. This is typical in oncology, where response in serious diseases justifies accepting some toxicity.
    \item $r = 1$ ($\delta = d$): Efficacy and safety are weighted equally, yielding $u_2 = u_3 = 0.5$.
    \item $r > 1$ ($\delta > d$): AEs are serious, so $u_3 > u_2$—avoiding toxicity becomes more important than achieving a response.
\end{itemize}

The ordering constraint $u_1 \geq u_2 \geq u_3 \geq u_4$ is automatically satisfied when $r \leq 1$. When $r > 1$ (safety margin smaller than efficacy margin), the margin-based approach for utility score yields $u_3 > u_2$, violating the assumed preference ranking. In this case, swapping the assignments of $u_2$ and $u_3$ restores the intended interpretation that higher utility scores represent better clinical outcomes. However, $r > 1$ is uncommon in oncology and often indicates that the chosen margins do not reflect the intended clinical priorities; revisiting the specifications is warranted before proceeding.

\paragraph{Alternative Specifications.}
The margin-based approach is parsimonious and clinically grounded, but the sample size methodology developed in subsequent sections is mathematically valid for any set of utility scores satisfying $u_1 \geq u_2 \geq u_3 \geq u_4$. Study teams may use direct elicitation or other methods if preferred; the consequences of relaxing the utility independence assumption are discussed in Section~\ref{subsubsec:sensitivity}.

\subsection{Dose Selection Rule}

The dose selection decision is based on a simple threshold comparison of the observed mean utility scores:
\begin{equation}
\text{Select Dose } H \text{ if } \bar{U}_H - \bar{U}_L > \lambda_u; \quad \text{otherwise, select Dose } L
\label{eq:selection_rule}
\end{equation}
Here, $\lambda_u \geq 0$ is a prespecified decision threshold representing the minimum required utility advantage for Dose $H$ to be selected over Dose $L$. 

When $\lambda_u = 0$, the rule simplifies to selecting the dose with the higher observed mean utility. However, a positive threshold ($\lambda_u > 0$) is warranted in two situations. First, the margin guards against selecting a dose based on random small differences arising from sampling variability. Secondly, it provides a formal mechanism to account for factors not captured in the utility score. For instance, a dose with higher manufacturing complexity, cost, or logistical burden must demonstrate a utility advantage of at least $\lambda_u$ to justify these additional hurdles, thereby promoting more robust and clinically relevant dose selection decisions.

Note that in the ROSE design, which relies solely on efficacy for dose selection, a positive threshold ($\lambda_u > 0$) is required to enforce a preference for the lower dose in the absence of a meaningful efficacy advantage, with the assumption that higher doses lead to higher risks. In contrast, our utility score method explicitly incorporates the efficacy-safety trade-off through the utility scores themselves, making a positive $\lambda_u$ optional rather than a requirement.

\section{Sample Size Determination}
\label{sec:sample_size}

\subsection{Clinical Scenarios for Sample Size Calculation}

To ensure reliable dose selection, the trial is sized to achieve a prespecified PCS under relevant clinical scenarios. A key feature of our framework is that the same clinical margins ($\delta$, $d$) used to derive the utility scores (Section~\ref{subsubsec:utility_specification}) also define the scenarios under which each dose should be selected. This coupling creates a design with strong internal consistency: the magnitude of benefit defined as clinically meaningful for constructing the utility scores is the same magnitude the trial is powered to detect. Consequently, any two doses with a true utility difference smaller than what is defined by these margins are treated as practically equivalent by the design. While this paper focuses on the parsimonious case where the utility and sample size margins are identical, the proposed methods are general and can easily accommodate different margins for each task if desired by the study team.

Using the margins from Equation~\eqref{eq:clinical_margins}, we define scenarios $\mathcal{S}_L$ and $\mathcal{S}_H$ as shown in Table~\ref{tab:selection_criteria01}.

\begin{table}[ht]
\centering
\caption{Decision Scenarios for Sample Size Calculation}
\label{tab:selection_criteria01}
\begin{tabular}{lccc}
\toprule
Scenario & Dose $L$: $(p_L, q_L)$ & Dose $H$: $(p_H, q_H)$ & PCS Requirement \\
\midrule
$\mathcal{S}_L$ & $(p, q)$ & $(p, q - d)$ & $\Pr(\text{select } L) \geq \alpha_L$ \\
$\mathcal{S}_H$ & $(p - \delta, q)$ & $(p, q)$ & $\Pr(\text{select } H) \geq \alpha_H$ \\
\bottomrule
\end{tabular}
\end{table}

\paragraph{Interpretation of Scenarios.}
Under Scenario $\mathcal{S}_L$, Dose $H$ offers no efficacy benefit ($p_H = p_L = p$) but has a safety disadvantage of magnitude $d$ (i.e., $q_H = q_L - d$, meaning a lower no-AE rate). In this case, the correct decision is to conservatively select Dose $L$, and the design must achieve this with probability at least $\alpha_L$. Under Scenario $\mathcal{S}_H$, Dose $H$ offers an efficacy advantage of magnitude $\delta$ (i.e., $p_H - p_L = \delta$) with equivalent safety ($q_H = q_L = q$). Here, the correct decision is to select Dose $H$, which must be achieved with probability at least $\alpha_H$.

\paragraph{Symmetric Utility Differences.}
A key property of the margin-based specification is that the utility advantage of the correct dose is symmetric across scenarios. Under $\mathcal{S}_L$, Dose $L$ is superior with $\mu_L - \mu_H = d \cdot u_3 = \delta/(1+r)$. Under $\mathcal{S}_H$, Dose $H$ is superior with $\mu_H - \mu_L = \delta \cdot u_2 = \delta/(1+r)$. This symmetry arises because the utility scores (Equation~\eqref{eq:utility_formula}) are derived from the same margins ($\delta$, $d$) that define the scenarios, ensuring that a $\delta$ efficacy advantage exactly compensates for a $d$ safety disadvantage. Consequently, both scenarios present identical statistical challenges for dose selection, naturally justifying symmetric PCS targets ($\alpha_L = \alpha_H$).

\subsection{Parameter Estimation from Historical Data}
\label{subsubsec:parameter_estimation}

The sample size calculation requires the following inputs, which can be estimated from historical data or elicited from clinical expertise:
\begin{itemize}
\item $p, q$: Desirable response rate and no-AE rate for the reference dose.
\item $\delta, d$: Clinically meaningful efficacy and safety margins.
\item $\phi$: Correlation between efficacy and safety endpoints.
\item Utility scores: $u_1, u_2, u_3, u_4$ (e.g., derived from $\delta/d$ via Equation~\eqref{eq:utility_formula}).
\item Target PCS levels: $\alpha_L, \alpha_H$.
\end{itemize}

The outcome probabilities $\boldsymbol{\pi} = (\pi_1, \pi_2, \pi_3, \pi_4)$ are fully determined by three quantities: the response rate $p = \Pr(X=1)$, the no-AE rate $q = \Pr(Y=1)$, and the Pearson correlation $\phi$ between $X$ and $Y$. From prior data with the counts $n_{11}, n_{10}, n_{01}, n_{00}$ of outcomes (as defined in Table~\ref{tab:utility01}), these can be estimated as:
\begin{equation}
\hat{p} = \frac{n_{11} + n_{10}}{n}, \quad \hat{q} = \frac{n_{11} + n_{01}}{n}, \quad
\hat{\phi} = \frac{n \cdot n_{11} - (n_{11}+n_{10})(n_{11}+n_{01})}{\sqrt{(n_{11}+n_{10})(n_{11}+n_{01})(n_{10}+n_{00})(n_{01}+n_{00})}}.
\label{eq:estimates}
\end{equation}

The outcome probabilities are then computed as:
\begin{equation}
\label{eq:pi_all}
\left\{
\begin{aligned}
    \pi_1 &= p q + \phi \sqrt{p(1-p) q(1-q)}, \\
    \pi_2 &= p - \pi_1, \\
    \pi_3 &= q - \pi_1, \\
    \pi_4 &= 1 - p - q + \pi_1.
\end{aligned}
\right.
\end{equation}

\paragraph{Constraints on the Correlation Parameter.}
The joint probability $\pi_1 = \Pr(X=1, Y=1)$ must satisfy the Fréchet-Hoeffding bounds to ensure all $\pi_k \geq 0$: $\max(0, p + q - 1) \leq \pi_1 \leq \min(p, q)$. This gives the achievable range for $\phi$:
\begin{equation}
\phi_{\min} = \frac{\max(0, p+q-1) - pq}{\sqrt{p(1-p)q(1-q)}}, \quad \phi_{\max} = \frac{\min(p,q) - pq}{\sqrt{p(1-p)q(1-q)}}.
\label{eq:phi_bounds}
\end{equation}
When estimating from data, $\hat{\phi}$ should be truncated to this range if necessary.

\paragraph{Moments of the Utility Score.}
With utility scores $\mathbf{u} = (u_1, u_2, u_3, u_4)^T$ and probability vector $\boldsymbol{\pi}$, a single patient's utility score $U = \mathbf{u}^T \mathbf{W}$ has moments:
\begin{equation}
\mu = \mathbb{E}[U] = \sum_{k=1}^4 u_k \pi_k, \quad \sigma^2 = \Var(U) = \mathbf{u}^T \left( \diag(\boldsymbol{\pi}) - \boldsymbol{\pi}\boldsymbol{\pi}^T \right) \mathbf{u}.
\label{eq:utility_moments}
\end{equation}

By the Central Limit Theorem, for large $n$, the sample mean utility $\bar{U}_j$ of $n$ patients is
\begin{equation}
\bar{U}_j \approx N\left( \mu_j, \frac{\sigma_j^2}{n} \right).
\end{equation}

Assuming independent samples for Dose $H$ and Dose $L$, the difference $\bar{D} = \bar{U}_H - \bar{U}_L$ has:
\begin{equation}
\bar{D} \approx N\left( \Delta\mu, \frac{\sigma_H^2 + \sigma_L^2}{n} \right),
\label{eq:utility_diff_dist}
\end{equation}
where $\Delta\mu = \mu_H - \mu_L$. This asymptotic normality forms the basis for the sample size calculations below.

\subsection{Approximation Method for Sample Size Calculation}

The sample size \(n\) per arm is determined to ensure that the dose optimization study meets prespecified PCS requirements across clinically plausible scenarios. This is achieved by finding the minimal sample size required to simultaneously satisfy the decision criteria for scenarios \(S_L\) and \(S_H\). This approach generalizes the optimization method of the ROSE design \citep{wang2025rose} to a utility-score framework.

Two PCS requirements in Table~\ref{tab:selection_criteria01} can be expressed as:
\begin{enumerate}
    \item \textbf{Under Scenario \(S_L\):} We require 
    \(\Pr(\bar{U}_H - \bar{U}_L \le \lambda_u \mid S_L) \ge \alpha_L\).
    \item \textbf{Under Scenario \(S_H\):} We require 
    \(\Pr(\bar{U}_H - \bar{U}_L > \lambda_u \mid S_H) \ge \alpha_H\).
\end{enumerate}

Let $\Delta\mu(S_L) < 0$ and $\Delta\mu(S_H) > 0$ be the true mean utility differences under scenarios $S_L$ and $S_H$, respectively. Expressing the PCS requirements in terms of standard normal quantiles yields:
\begin{align}
   \text{Under } S_L: \quad \lambda_u 
   &\ge \Delta\mu(S_L) + z_{\alpha_L} \sqrt{\frac{v(S_L)}{n}}, 
   \label{eq:lambda_lower} \\
   \text{Under } S_H: \quad \lambda_u 
   &\le \Delta\mu(S_H) + z_{1-\alpha_H} \sqrt{\frac{v(S_H)}{n}}, 
   \label{eq:lambda_upper}
\end{align}
where $v(S) = \sigma_H^2(S) + \sigma_L^2(S)$ is the sum of variances under scenario $S$. For any threshold $\lambda_u \in (\Delta\mu(S_L),\, \Delta\mu(S_H))$, it is always possible to find a sample size $n$ large enough to satisfy both requirements simultaneously.

\subsubsection{Direct Approach: Sample Size for a Given Threshold}

For a prespecified threshold $\lambda_u$, each PCS requirement can be solved independently for~$n$ by rearranging Equations~\eqref{eq:lambda_lower} and \eqref{eq:lambda_upper}:
\begin{align}
n_L(\lambda_u) &= \left\lceil 
\frac{z_{\alpha_L}^2 \, v(S_L)}{(\lambda_u - \Delta\mu(S_L))^2} 
\right\rceil, 
\label{eq:n_L_direct} \\[6pt]
n_H(\lambda_u) &= \left\lceil 
\frac{z_{\alpha_H}^2 \, v(S_H)}{(\Delta\mu(S_H) - \lambda_u)^2} 
\right\rceil.
\label{eq:n_H_direct}
\end{align}
The overall required sample size for a given threshold is $n(\lambda_u) = \max\bigl(n_L(\lambda_u),\; n_H(\lambda_u)\bigr)$. Each formula has the familiar structure of a one-sided Z-test sample size detecting a shift of $|\lambda_u - \Delta\mu(S)|$ in a normal mean with variance $v(S)$. This direct approach is useful when clinical convention or regulatory guidance dictates a specific threshold (e.g., $\lambda_u = 0$), and it provides transparency into the relative difficulty of the two scenarios by revealing which is the binding constraint.

\subsubsection{Optimal Threshold: Minimizing Sample Size}

Since $n_L(\lambda_u)$ is monotonically decreasing and $n_H(\lambda_u)$ is monotonically increasing in $\lambda_u$, the minimum of the maximum of $n_L(\lambda_u)$ and $n_H(\lambda_u)$ is achieved at the intersection $n_L(\lambda_u^*) = n_H(\lambda_u^*)$. Setting Equations~\eqref{eq:n_L_direct} and \eqref{eq:n_H_direct} equal (without ceiling functions) and solving yields the jointly optimal sample size:
\begin{equation}
\boxed{n = \left[ \frac{z_{\alpha_L} \sqrt{v(S_L)} 
- z_{1-\alpha_H} \sqrt{v(S_H)}}
{\Delta\mu(S_H) - \Delta\mu(S_L)} \right]^2.}
\label{eq:sample_size}
\end{equation}
The corresponding optimal threshold is:
\begin{equation}
\lambda_u^* = \Delta\mu(S_H) + z_{1-\alpha_H} 
\sqrt{\frac{v(S_H)}{n}}.
\label{eq:optimal_threshold}
\end{equation}
To ensure robust performance, study teams should evaluate several plausible scenarios and select the largest resulting sample size.

\subsection{Exact Calculation for Sample Sizes}

For small sample sizes, the normal approximation may be inaccurate. We develop an exact calculation approach using the Multinomial probability mass function (PMF) directly.

Let $\mathbf{n}_L = (n_{11,L}, n_{10,L}, n_{01,L}, n_{00,L})$ and $\mathbf{n}_H = (n_{11,H}, n_{10,H}, n_{01,H}, n_{00,H})$ denote the count vectors for Dose $L$ and Dose $H$, respectively, where each vector sums to $n$.

Under independence between doses, the joint probability is:
\[
P(\mathbf{n}_L, \mathbf{n}_H) = \Mult(\mathbf{n}_L \mid n, \boldsymbol{\pi}_L) \cdot \Mult(\mathbf{n}_H \mid n, \boldsymbol{\pi}_H).
\]

The exact PCS under Scenario \(S_L\) is:
\begin{equation}
\text{PCS}_L(\lambda_u; n) = \sum_{\mathbf{n}_L, \mathbf{n}_H} \mathbb{I}(\bar{U}_H - \bar{U}_L \leq \lambda_u) \cdot \Mult(\mathbf{n}_L \mid n, \boldsymbol{\pi}_L^{S_L}) \cdot \Mult(\mathbf{n}_H \mid n, \boldsymbol{\pi}_H^{S_L}),
\label{eq:exact_pcs_L}
\end{equation}
and similarly for Scenario \(S_H\):
\begin{equation}
\text{PCS}_H(\lambda_u; n) = \sum_{\mathbf{n}_L, \mathbf{n}_H} \mathbb{I}(\bar{U}_H - \bar{U}_L > \lambda_u) \cdot \Mult(\mathbf{n}_L \mid n, \boldsymbol{\pi}_L^{S_H}) \cdot \Mult(\mathbf{n}_H \mid n, \boldsymbol{\pi}_H^{S_H}).
\label{eq:exact_pcs_H}
\end{equation}

The minimal sample size is found via grid search over \(n\) and \(\lambda_u\) to find the smallest \(n\) such that:
\begin{equation}
\text{PCS}_L(\lambda_u; n) \geq \alpha_L \quad \text{and} \quad \text{PCS}_H(\lambda_u; n) \geq \alpha_H.
\label{eq:exact_sample_size}
\end{equation}

\subsection{Mathematical Justification of the Margin-Based Approach}
\label{subsec:math_justification}

This section provides a rigorous mathematical justification for the utility score specification in Section~\ref{subsubsec:utility_specification}. 

\subsubsection{Marginal Rate of Substitution}

The trade-off between efficacy and safety can be characterized through the \textit{marginal rate of substitution} (MRS)---the increase in efficacy ($\Delta p > 0$) required to compensate for a decrease in safety ($\Delta q < 0$) while maintaining constant mean utility $\mu$:
\begin{equation}
\text{MRS} = \frac{\Delta p}{|\Delta q|} = \frac{\partial \mu / \partial q}{\partial \mu / \partial p}.
\label{eq:mrs_def}
\end{equation}

Substituting Equation~\eqref{eq:pi_all} into Equation~\eqref{eq:utility_moments} and simplifying, the mean utility can be written as:
\begin{equation}
\mu = (u_1 - u_2 - u_3 + u_4) [pq + \phi \sqrt{p(1-p) q(1-q)}]
+ (u_2 - u_4)p + (u_3 - u_4)q + u_4.
\label{eq:mu_expanded}
\end{equation}

The partial derivatives are:
\begin{align}
\frac{\partial \mu}{\partial p}
&= (u_1 - u_2 - u_3 + u_4)
\left[ q + \phi \cdot \frac{1}{2}
\sqrt{\frac{q(1-q)}{p(1-p)}} (1 - 2p) \right] + (u_2 - u_4), \label{eq:dmu_dp} \\
\frac{\partial \mu}{\partial q}
&= (u_1 - u_2 - u_3 + u_4)
\left[ p + \phi \cdot \frac{1}{2}
\sqrt{\frac{p(1-p)}{q(1-q)}} (1 - 2q) \right] + (u_3 - u_4). \label{eq:dmu_dq}
\end{align}

\subsubsection{Simplification Under Utility Independence}

Under the utility independence assumption ($u_1 - u_2 - u_3 + u_4 = 0$), the partial derivatives simplify dramatically:
\begin{equation}
\frac{\partial \mu}{\partial p} = u_2 - u_4, \quad \frac{\partial \mu}{\partial q} = u_3 - u_4.
\end{equation}

The MRS becomes \textbf{independent of $p$, $q$, and $\phi$}:
\begin{equation}
\text{MRS} = \frac{u_3 - u_4}{u_2 - u_4}.
\label{eq:mrs_simplified}
\end{equation}

Recall from Section~\ref{subsubsec:utility_specification} that the clinical margins $\delta$ and $d$ each represent one unit of clinical significance, so a $\delta$ efficacy gain should exactly offset a $d$ safety loss. This implies MRS $= \delta/d = r$. Substituting into Equation~\eqref{eq:mrs_simplified} with $u_4 = 0$ and the utility independence condition $u_2 + u_3 = 1$, we have:
\[
\frac{u_3}{u_2} = r \quad \Rightarrow \quad u_2 = \frac{1}{1+r}, \quad u_3 = \frac{r}{1+r},
\]
which recovers Equation~\eqref{eq:utility_formula}. Therefore, the clinically motivated derivations of utility score in Section~\ref{subsubsec:utility_specification} are uniquely determined by the requirement of a constant MRS equal to the trade-off ratio $r= \delta/d$ under utility independence.

\subsubsection{Sensitivity to the Utility Independence Assumption}
\label{subsubsec:sensitivity}

The margin-based utility specification relies on the utility independence assumption (Equation~\ref{eq:utility_independence}), which implies $u_1 - u_2 - u_3 + u_4 = 0$, or equivalently, $u_2 + u_3 = 1$ when $u_1 = 1$ and $u_4 = 0$. To understand these effects, recall that the mean utility $\mu$ for a dose (Equation~\eqref{eq:mu_expanded}) can be decomposed as:
\begin{equation}
\label{eq:mu_2}
 \mu = \eta \cdot \pi_1 + (u_2 - u_4) p + (u_3 - u_4) q + u_4,   
\end{equation}
where $\eta = u_1 - u_2 - u_3 + u_4$, and $\pi_1 = p q + \phi \sqrt{p(1-p) q(1-q)}$ is the joint probability of achieving both response ($X=1$) and no adverse event ($Y=1$).

Under utility independence ($\eta = 0$, or $u_2 + u_3 = 1$), the mean utility simplifies to a weighted average: $\mu = (u_2 - u_4) p + (u_3 - u_4) q + u_4$. Here, efficacy ($p$) and safety ($q$) contribute additively, with no amplification or dampening from their joint occurrence.

When utility independence is relaxed, the MRS (Equation~\eqref{eq:mrs_def}) becomes dependent on the rates $(p, q)$ and the correlation $\phi$, causing the trade-off ratio to vary across the parameter space, rather than remaining constant at $r = \delta/d$. Additionally, the term $\eta \cdot \pi_1$ introduces an interaction effect in Equation~\eqref{eq:mu_2}:
\begin{itemize}
    \item $\eta > 0$ ($u_2 + u_3 < 1$): Positive interaction—joint success (response + no AE) is ``synergistically valuable,'' boosting $\mu$ more than the sum of individual contributions. This increases both $|\Delta\mu(\mathcal{S}_H)|$ and $|\Delta\mu(\mathcal{S}_L)|$, leading to \textit{smaller} required sample sizes.
    \item $\eta < 0$ ($u_2 + u_3 > 1$): Negative interaction—joint success is undervalued relative to the sum, dampening $\mu$. This decreases both $|\Delta\mu(\mathcal{S}_H)|$ and $|\Delta\mu(\mathcal{S}_L)|$, leading to \textit{larger} required sample sizes.
\end{itemize}

While the sample size methodology developed in this paper is mathematically valid for any utility scores satisfying $u_1 \geq u_2 \geq u_3 \geq u_4$, we recommend the margin-based approach with utility independence for most applications. This specification offers three advantages: (1) it requires only two clinically interpretable inputs ($\delta$, $d$); (2) it guarantees a constant MRS that directly reflects the intended trade-off; and (3) it ensures symmetric utility differences across scenarios (i.e., $|\Delta\mu(\mathcal{S}_L)| = |\Delta\mu(\mathcal{S}_H)|$), which provides a natural justification for setting equal PCS targets ($\alpha_L = \alpha_H$) since both scenarios present the same magnitude of utility difference to detect.

\subsection{ROSE Design as a Special Case}
\label{subsec:rose_special_case}

To demonstrate that our framework generalizes the one-stage option of ROSE design, we consider a special case where utility is defined solely by response: \(u_1 = u_2 = 1\) and \(u_3 = u_4 = 0\). In this setting, the mean utility \(\mu_j\) simplifies to the response rate \(p_j\), and the variance \(\sigma_j^2\) becomes \(p_j(1-p_j)\).

Using the ROSE design parameters from Table 1 of \citet{wang2025rose} (\(p_H = 0.4\), \(\delta = 0.15\), corresponding to \(p_L = 0.25\) in \(S_H\)):
\begin{itemize}
    \item Under \(S_L\): \(\Delta\mu(S_L) = 0\) and \(v(S_L) = 2 \times 0.4 \times 0.6 = 0.48\),
    \item Under \(S_H\): \(\Delta\mu(S_H) = 0.15\) and \(v(S_H) = 0.4 \times 0.6 + 0.25 \times 0.75 = 0.4275\).
\end{itemize}
Solving Equation \eqref{eq:sample_size} with \(\alpha_L = \alpha_H = 0.8\) yields \(n = 58\) per arm and \(\lambda_u \approx 0.077\). These results \textbf{exactly match} those reported in the original ROSE design (Table 1 of \citealp{wang2025rose}), confirming that our utility-based framework fully encompasses ROSE when safety is ignored.

\section{Selection-Induced Bias Characterization}
\label{sec:bias}

The dose optimization design, which selects the dose with the higher observed mean utility, can introduce a positive selection-induced bias in the estimate of the chosen dose's efficacy. This bias is most pronounced under the null hypothesis, where the two doses have identical true efficacy and safety profiles.

The selection-induced bias directly affects the interpretation of the point estimate of the selected dose's efficacy at the end of the dose optimization study, potentially leading to over-optimistic assessments. The understanding of selection-bias is also crucial for planning future studies because biased estimates from the optimization stage may result in incorrect hypotheses or assumptions being used in sample size calculations, thereby compromising the power or validity of subsequent trials. Additionally, if data from the dose optimization study are combined with data from a future confirmatory trial (a common practice to maximize statistical power), the bias will propagate to the final analysis, potentially inflating the Type I error rate.

Although this paper focuses on one-stage dose optimization, the selected dose will typically advance to a Phase 3 registration study. The following subsections characterize the bias and its downstream consequences for both binary response rate and time-to-event confirmatory endpoints, enabling sponsors to make informed decisions about trial design parameters and data pooling strategies.

\subsection{Bias in Response Rate due to Dose Selection}
\subsubsection{Bias in Response Rate Estimation}

The bias is defined as the difference between the expected response rate of the selected dose and the true response rate:
\[
\Bias = \mathbb{E}[\hat{p}_{\text{selected}}] - p,
\]
where \(\hat{p}_{\text{selected}}\) is the observed response rate of the chosen dose. 
Under the null hypothesis ($p_L = p_H = p$, $q_L = q_H = q$), we derive the bias as:
\begin{equation}
\boxed{\Bias = \frac{\Cov(X, U)}{\sigma_U \sqrt{n}} \cdot \frac{1}{\sqrt{\pi}} \exp\left( -\frac{\lambda_u^2 n}{4 \sigma_U^2} \right),}
\label{eq:bias_final}
\end{equation}
where \(\sigma^2_U = \Var(U)\) is the variance of the utility score \(U\), $\Cov(X, U) = E[XU] - E[X]E[U] = \sum_{i=1}^4 x_k u_k \pi_k - p \mu$, and $\pi$ in the denominator denotes the mathematical constant ($\pi \approx 3.14159$), not a probability parameter.

The bias is maximized when selection is based solely on efficacy (\(\Cov(X, U) = \sqrt{\Var(X) \sigma_U^2}\), as in the ROSE design with \(u_1 = u_2 = 1\), \(u_3 = u_4 = 0\)) and when \(\lambda_u = 0\) (no additional hurdle for dose selection). In this case, the exponential term becomes 1, yielding:
\begin{equation}
\Bias_{\max} = \frac{\sqrt{p(1-p)}}{\sqrt{n \pi}}.
\label{eq:bias_max}
\end{equation}

\subsubsection{Derivation of the Bias Formula}
\label{subsec:bias_derivation}

Under the null hypothesis (\(p_L = p_H = p\), \(q_L = q_H = q\)), the mean utility scores 
for both doses are independent and identically distributed, with 
\(\bar{U}_H, \bar{U}_L \overset{\text{iid}}{\sim} N(\mu, \sigma_U^2 / n)\), where \(\mu = \mathbb{E}[U]\).

The utility difference \(D = \bar{U}_H - \bar{U}_L\) follows \(D \sim N(0, 2\sigma_U^2 / n)\). 
The response rate of the selected dose is:
\begin{equation}
\hat{p}_{\text{selected}} =
\begin{cases}
\hat{p}_H, & \text{if } D > \lambda_u \\
\hat{p}_L, & \text{if } D \le \lambda_u
\end{cases}
\end{equation}
where, under the null, $E[\hat{p}_H]=E[\hat{p}_L]=p$. 

By symmetry under the null:
\begin{equation}
\mathbb{E}[\hat{p}_{\text{selected}}] - p = \mathbb{E}\left[ (\hat{p}_H - p) \cdot \mathbf{1}\{D > \lambda_u\} \right] 
+ \mathbb{E}\left[ (\hat{p}_L - p) \cdot \mathbf{1}\{D \le \lambda_u\} \right]. 
\end{equation}

Under the Central Limit Theorem, \((\hat{p}_j, \bar{U}_j)\) are approximately jointly normal. 
For jointly normal random variables, the conditional expectation is:
\begin{equation}
\mathbb{E}[\hat{p}_j - p \mid \bar{U}_j] = \frac{\Cov(\hat{p}_j, \bar{U}_j)}{\Var(\bar{U}_j)}(\bar{U}_j - \mu) 
= \frac{\Cov(X, U)}{\sigma_U^2}(\bar{U}_j - \mu),
\end{equation}
where the simplification uses \(\Cov(\hat{p}_j, \bar{U}_j) = \Cov(X, U)/n\) and \(\Var(\bar{U}_j) = \sigma_U^2/n\).

Since the selection event \(\{D > \lambda_u\}\) is determined by \((\bar{U}_H, \bar{U}_L)\), 
we apply the law of iterated expectations:
\begin{equation}
\label{eq:average_high}
\mathbb{E}\left[ (\hat{p}_H - p) \cdot \mathbf{1}\{D > \lambda_u\} \right] 
= \mathbb{E}\left[ \mathbb{E}[\hat{p}_H - p \mid \bar{U}_H] \cdot \mathbf{1}\{D > \lambda_u\} \right]
= \frac{\Cov(X, U)}{\sigma_U^2} \mathbb{E}\left[ (\bar{U}_H - \mu) \cdot \mathbf{1}\{D > \lambda_u\} \right].
\end{equation}

Define standardized variables \(Z_j = (\bar{U}_j - \mu)/(\sigma_U / \sqrt{n})\), which are i.i.d.\ \(N(0,1)\) for $j \in \{H, L\}$. By substituting \(\bar{U}_j - \mu = (\sigma_U/\sqrt{n}) Z_j\) into the expression for the high-dose arm in Equation~\eqref{eq:average_high}, we get:
\begin{equation}
\mathbb{E}\left[ (\hat{p}_H - p) \cdot \mathbf{1}\{D > \lambda_u\} \right] 
= \frac{\Cov(X, U)}{\sigma_U\sqrt{n}} \mathbb{E}\left[ Z_H \cdot \mathbf{1}\{D > \lambda_u\} \right].
\end{equation}
Similar expression holds for the low-dose term. Combining both terms yields:
\begin{equation}
    \Bias =\mathbb{E}[\hat{p}_{\text{selected}}] - p = \frac{\Cov(X, U)}{\sigma_U \sqrt{n}} \mathbb{E}\left[ Z_H \cdot \mathbf{1}\{D > \lambda_u\} + Z_L \cdot \mathbf{1}\{D \le \lambda_u\} \right].    
\end{equation}

Since \(D = (\sigma_U/\sqrt{n})(Z_H - Z_L)\), the selection rule \(D > \lambda_u\) is equivalent to \(Z_H - Z_L > k\), where \(k = \lambda_u\sqrt{n}/\sigma_U\). By Lemma~\ref{lem:truncated}, the expectation evaluates to:
\begin{equation}
    \mathbb{E}\left[ Z_H \cdot \mathbf{1}\{Z_H - Z_L > k\} 
    + Z_L \cdot \mathbf{1}\{Z_H - Z_L \le k\} \right] 
    = \frac{1}{\sqrt{\pi}} e^{-k^2/4}.   
\end{equation}

Substituting $k^2/4 = \lambda_u^2 n / (4\sigma_U^2)$ yields the final result in Equation~\eqref{eq:bias_final}. When $\lambda_u = 0$ (equivalently, $k = 0$), this reduces to $\mathbb{E}[\max(Z_1, Z_2)] = 1/\sqrt{\pi}$, recovering the classical result that \citet{bauer2010selection} used to derive selection bias $m_1(2)\sigma_U/\sqrt{n}$ for two treatments with equal means. The exponential term $e^{-k^2/4}$ extends bias calculation to threshold-based selection, while the $\Cov(X, U)$ term in Equation~\eqref{eq:bias_final} generalizes to settings where the selection criterion and estimation use different variables. 

The exponential term $\exp(-\lambda_u^2 n / 4\sigma_U^2)$ reflects the mechanism by which thresholding reduces selection bias. Without a threshold ($\lambda_u = 0$), the bias arises from always selecting the dose with the higher observed utility, which preferentially selects doses that have experienced favorable random fluctuations. When a positive threshold is imposed, selection occurs only when the utility difference exceeds $\lambda_u$, corresponding to the tails of the distribution where the conditional expectation $\mathbb{E}[Z_H\mid D>\lambda_u]$ is largest. As $\lambda_u$ increases, the probability mass in these tails decreases exponentially, thereby attenuating the bias.

\subsection{Implications for Confirmatory Trial Planning: Type I Error Inflation}
\label{sec:type_one_error_binary}

To illustrate the practical impact of selection bias on subsequent studies, we consider a hypothetical two-stage clinical development pathway where the dose optimization study (Stage~1 with $n_1$ patients per dose) is followed by a confirmatory trial (Stage~2 with $n_2$ newly enrolled patients on the selected dose). This two-stage framing is used solely to characterize how bias propagates to downstream inference; it does not represent a proposed two-stage design.

If response rate data from both stages is pooled, the combined response rate across both stages is:
\begin{equation}
\hat{p}_{\text{combined}} = \frac{n_1 \hat{p}_{\text{selected}} + n_2 \hat{p}_2}{n_1 + n_2},
\label{eq:combined_response}
\end{equation}
where \(\hat{p}_2\) is the response rate observed in Stage 2. Under the null hypothesis, the expected combined response rate is:
\begin{equation}
\mathbb{E}[\hat{p}_{\text{combined}}] = \frac{n_1}{n_1 + n_2}\mathbb{E}[\hat{p}_{\text{selected}}] + \frac{n_2}{n_1 + n_2} p.
\end{equation}

Thus, the bias in the combined estimate is:
\begin{equation}
\Delta p_{\text{combined}} = \frac{n_1}{n_1 + n_2} \cdot \Bias,
\label{eq:combined_bias}
\end{equation}
where Bias is given by Equation~\eqref{eq:bias_final}. This shows that the bias is diluted by the confirmatory stage, with the dilution factor \(\frac{n_1}{n_1 + n_2}\). The maximum bias occurs at the end of the dose optimization stage when \(n_2 = 0\).

\subsubsection{Type I Error Inflation for Z-test}

Consider testing the null hypothesis \(H_0: p \le p_0\) at a nominal one-sided significance level \(\alpha\) using the Z-test statistic:
\begin{equation}
Z = \frac{\hat{p}_{\text{combined}} - p_0}{\SE_0},
\end{equation}
where \(\SE_0 = \sqrt{p_0(1-p_0)/(n_1 + n_2)}\) is the standard error under the null hypothesis. The null hypothesis is rejected when \(Z > z_{1-\alpha}\), where \(z_{1-\alpha}\) is the \((1-\alpha)\) quantile of the standard normal distribution.

Because of the positive selection-induced bias, the actual distribution of \(\hat{p}_{\text{combined}}\) under the null is shifted by \(\Delta p_{\text{combined}}\), leading to a change to the test statistics:
\begin{equation}
    \Bias(Z_{\text{Binary}}) = \frac{\Delta p_{\text{combined}}}{\SE_0}.    
\end{equation}

Consequently, the inflated Type I error rate for the Z-test becomes:
\begin{equation}
    \text{Type I Error}_{\text{Z}} = \Pr\left(Z > z_{1-\alpha} \mid p = p_0\right) = 1 - \Phi\left(z_{1-\alpha} - \Bias(Z_{\text{Binary}})\right),
\label{eq:type_one_error_z}
\end{equation}
where \(\Phi(\cdot)\) denotes the standard normal cumulative distribution function.

Substituting the utility-based bias from Equation~\eqref{eq:bias_final} yields:
\begin{equation}
    \Delta p_{\text{combined}} = \frac{n_1}{n_1 + n_2} \cdot \frac{\Cov(X, U)}{\sigma_U \sqrt{n_1}} \cdot \frac{1}{\sqrt{\pi}} \exp\left( -\frac{\lambda_u^2 n_1}{4 \sigma_U^2} \right).
\label{eq:delta_p_utility}
\end{equation}

Using the maximum bias from Equation~\eqref{eq:bias_max}, we obtain an upper bound:
\begin{equation}
\Delta p_{\text{max,combined}} = \frac{n_1}{n_1 + n_2} \cdot \frac{\sqrt{p_0(1-p_0)}}{\sqrt{n_1 \pi}},
\label{eq:max_combined_bias}
\end{equation}
and the corresponding upper bound on Type I error for the Z-test:
\begin{equation}
\text{Type I Error}_{\text{Z,max}} = 1 - \Phi\left(z_{1-\alpha} - \Bias(Z_{\text{max,Binary}})\right).
\label{eq:type_one_error_z_max}
\end{equation}
where 
\begin{equation}
   \Bias(Z_{\text{max,Binary}}) = \frac{\Delta p_{\text{max,combined}}}{\SE_0}. 
\end{equation}

\subsubsection{Type I Error Inflation for Binomial Test}

For an exact binomial test, let \(k_c\) denote the critical value determined under the null hypothesis such that:
\begin{equation}
\Pr(X > k_c \mid n_1 + n_2, p_0) \le \alpha,
\end{equation}
where \(X\) is the total number of responders. Due to the positive selection-induced bias, the true success probability under the null is effectively shifted to \(p_0 + \Delta p_{\text{combined}}\). The inflated Type I error rate for the binomial test is:
\begin{equation}
\text{Type I Error}_{\text{Bin}} = \Pr(X > k_c \mid n_1 + n_2, p_0 + \Delta p_{\text{combined}}) = 1 - F_{n_1+n_2,\, p_0 + \Delta p_{\text{combined}}}(k_c),
\label{eq:type_one_error_bin}
\end{equation}
where \(F_{n,p}(\cdot)\) denotes the cumulative distribution function of the binomial distribution with sample size \(n\) and success probability \(p\).

Similarly, using the maximum bias bound yields:
\begin{equation}
\text{Type I Error}_{\text{Bin,max}} = 1 - F_{n_1+n_2,\, p_0 + \Delta p_{\text{max,combined}}}(k_c).
\label{eq:type_one_error_bin_max}
\end{equation}
These analytical expressions provide study teams with tools to anticipate Type I error inflation when planning confirmatory trials that pool data from preceding dose optimization study.

\subsection{Bias Propagation to Time-to-Event Confirmatory Endpoints}
\label{subsec:bias_tte}

In many oncology development programs, the dose optimization study (Stage~1) relies on short-term binary endpoints such as objective response rate (ORR), whereas the confirmatory trial (Stage~2) uses time-to-event (TTE) endpoints such as progression-free survival (PFS) or overall survival (OS). When response and survival are positively correlated---as is typical in oncology, where responders tend to have longer survival---the selection-induced bias in the response rate can propagate to the TTE endpoint. The magnitude of this propagation depends on the strength of the response--survival association.

\subsubsection{Landmark Survival Test}
\label{subsubsec:landmark_test}

\paragraph{Bias.}
Since the $\tau$-month survival indicator $S(\tau) = \mathbf{1}\{T > \tau\}$ is a binary endpoint with the same algebraic structure as the response indicator $X$, the general bias formula (Theorem~\ref{thm:bias_general}) applies directly with $X = S(\tau)$:
\begin{equation}
\Bias(\hat{S}(\tau)) = \frac{\Cov(S(\tau), U)}{\sigma_U \sqrt{n_1}} \cdot \frac{1}{\sqrt{\pi}}\exp\!\left(-\frac{\lambda_u^2 n_1}{4\sigma_U^2}\right),
\label{eq:bias_landmark}
\end{equation}

where $\Cov(S(\tau), U)$ depends on the joint distribution of survival at time $\tau$ and the utility score. In practice, $\Cov(S(\tau), U)$ is estimated directly from Stage~1 data for the selected dose, providing a plugin estimate that requires no parametric assumptions about the response-survival relationship. Under a positive association between response and survival, $\Cov(S(\tau), U) > 0$ leads to an overestimation of the landmark survival probability.

An upper bound on the landmark bias can be obtained by noting that $|\Cov(S(\tau), U)| \leq \sigma_{S(\tau)} \cdot \sigma_U$, yielding:
\begin{equation}
\Bias_{\max}(\hat{S}(\tau))
= \frac{\sqrt{S_0(\tau)(1-S_0(\tau))}}{\sqrt{n_1 \pi}}
  \cdot \exp\!\left(-\frac{\lambda_u^2 n_1}{4\sigma_U^2}\right).
\label{eq:tte_upper}
\end{equation}

When response and survival are independent, $\Cov(S(\tau), U) = 0$ and the landmark bias vanishes entirely, regardless of the selection bias in the response rate.

\paragraph{Type~I Error.}
For a one-sided Z-test of $H_0\!: S(\tau) \leq S_0$, the test statistic is:
\begin{equation}
Z_{\text{Landmark}}
= \frac{\hat{S}(\tau)_{\text{comb}} - S_0}{\SE_0},
\label{eq:z_landmark}
\end{equation}
where $\SE_0 = \sqrt{S_0(1-S_0)/(n_1+n_2)}$. After applying the stage-combination dilution factor $n_1/(n_1+n_2)$ from Equation~\eqref{eq:combined_bias}, the bias in the test statistic is:
\begin{equation}
\Bias(Z_{\text{Landmark}})
= \frac{\Bias(\hat{S}(\tau))_{\text{comb}}}{\SE_0}.
\label{eq:bias_z_landmark}
\end{equation}
The inflated Type~I error follows from the shifted normal model:
\begin{equation}
\text{Type~I Error}_{\text{Landmark}}
= 1 - \Phi\!\left(z_{1-\alpha}
  - \Bias(Z_{\text{Landmark}})\right).
\label{eq:type1_landmark}
\end{equation}

\subsubsection{One-sample Exponential Test}
\label{subsubsec:exp_test}

\paragraph{Bias.}
For a one-sample test comparing the selected dose to a historical control hazard $\lambda_0$, the test operates on the log-hazard scale. The bias propagates through the delta method chain: from the mean survival time $\bar{T}$ (whose bias follows Theorem~\ref{thm:bias_general} with $W = T$) to the hazard rate and then to the log-hazard (derivations in Appendix~\ref{app:bias_tte_derivations}):
\begin{align}
\Bias(\hat{\lambda})
  &\approx -\lambda_0^2 \cdot \Bias(\bar{T}),
  \label{eq:bias_hazard} \\
\Bias(\log\hat{\lambda})_{\text{comb}}
  &= \frac{\Bias(\hat{\lambda})_{\text{comb}}}{\lambda_0}
   = -\lambda_0 \cdot \frac{n_1}{n_1+n_2} \cdot \Bias(\bar{T}),
  \label{eq:bias_log_hazard_comb}
\end{align}
where $\Bias(\bar{T})$ follows Theorem~\ref{thm:bias_general} with $W = T$ and is estimated from Stage~1 data via $\Cov(T, U)$. The approximation symbol $\approx$ in the bias formulas arises from the delta method chain. The negative sign indicates that positive bias in mean survival translates to an underestimated hazard rate: the selected dose appears to have better survival than reality.

\paragraph{Type~I Error.}
The MLE $\hat{\lambda} = D/\sum T_i$ satisfies $\Var(\log\hat{\lambda}) = 1/D$ with $D$ the expected number of events under the null. The test statistic is:
\begin{equation}
Z_{\text{Exp}}
= (\log\hat{\lambda} - \log\lambda_0)\sqrt{D}.
\label{eq:z_exp}
\end{equation}
The test rejects $H_0\!: \lambda \geq \lambda_0$ when $Z_{\text{Exp}} \leq -z_{1-\alpha}$. The bias in the test statistic is:
\begin{equation}
\Bias(Z_{\text{Exp}})
= \Bias(\log\hat{\lambda})_{\text{comb}} \cdot \sqrt{D}.
\label{eq:bias_z_exp}
\end{equation}
Since $\Bias(\bar{T}) > 0$ implies $\Bias(Z_{\text{Exp}}) < 0$, the inflated Type~I error is:
\begin{equation}
\text{Type~I Error}_{\text{Exp}}
= \Phi\!\left(-z_{1-\alpha} - \Bias(Z_{\text{Exp}})\right).
\label{eq:type1_exp}
\end{equation}

\subsubsection{Two-sample Wald Test for Cox Regression}
\label{subsubsec:cox_test}

\paragraph{Bias.}
For a two-sample Cox test with 1:1 randomization, the log hazard ratio is $\hat{\beta} = \log(\hat{\lambda}_{\text{trt}} / \hat{\lambda}_{\text{ctrl}})$. Since the control arm is unaffected by dose selection, the combined-stage bias is (Appendix~\ref{app:bias_tte_derivations}):
\begin{equation}
\operatorname{Bias}(\hat{\beta})_{\operatorname{comb}}
= -\lambda_0 \cdot \frac{n_1}{n_1+n_2} \cdot \Bias(\bar{T}).
\label{eq:bias_beta_comb}
\end{equation}
Note that this has the same functional form as $\Bias(\log\hat{\lambda})_{\text{comb}}$ in Equation~\eqref{eq:bias_log_hazard_comb}, since both arise from the delta method transformation $h(\lambda) = \log(\lambda)$ around $\lambda_0$. Under the Cox model, since $\mathbb{E}[\hat{\beta}] = 0$ in the absence of selection bias, the following is exact under the null hypothesis:
\begin{equation}
\label{eq:beta_bias}
\mathbb{E}[\hat{\beta}]
= \operatorname{Bias}(\hat{\beta})_{\mathrm{comb}}.
\end{equation}

\paragraph{Type~I Error.}
The Cox Wald statistic is:
\begin{equation}
Z_{\text{Cox}}
= \hat{\beta} \cdot \sqrt{\frac{D_{\text{total}}}{4}},
\label{eq:z_cox_wald}
\end{equation}
where $D_{\text{total}}$ is the total number of events across both arms and $\mathcal{I} = D_{\text{total}}/4$ is the Fisher information under equal allocation. The test rejects $H_0$ when $Z_{\text{Cox}} \leq -z_{1-\alpha}$. By linearity of expectation and Equation~\eqref{eq:beta_bias}, the bias in the test statistic is:
\begin{equation}
\Bias(Z_{\text{Cox}})
= \operatorname{Bias}(\hat{\beta})_{\operatorname{comb}}
  \cdot \sqrt{\frac{D_{\text{total}}}{4}},
\label{eq:bias_z_cox}
\end{equation}
which is exact given the asymptotic normality of $\hat{\beta}$ (Proposition~\ref{app:tte_type1_proofs}). The Type~I error follows from the lower-tail rejection:
\begin{equation}
\operatorname{Type~I~Error}_{\text{Cox}}
= \Phi\!\left(-z_{1-\alpha} - \Bias(Z_{\text{Cox}})\right).
\label{eq:type1_cox_main}
\end{equation}
Since $\Bias(\bar{T}) > 0$ implies $\Bias(Z_{\text{Cox}}) < 0$, the Type~I error exceeds the nominal level $\alpha$.

\section{Simulation Studies}
\label{sec:simulations}

We conducted extensive Monte Carlo simulations to evaluate the operating characteristics of the proposed utility-based selection framework. These simulations assessed three key properties: the validity of the sample size calculations, the magnitude of estimation bias, and the inflation of Type I error rates.

\subsection{Simulation Settings}

\subsubsection{Correlated Data Generation}
\label{subsubsec:sim_datagen}

Each simulated patient has three endpoints: binary efficacy $X_i \sim \text{Bernoulli}(p)$, binary safety $Y_i \sim \text{Bernoulli}(q)$, and (when applicable) a latent survival time $T_i \sim \text{Exp}(\lambda)$. The dependence structure is constructed in two steps.

First, the association between efficacy and survival is induced via a Gaussian copula with parameter $\rho_c$, which controls the strength of the efficacy--survival correlation: $\rho_c = 0$ yields independence, while $\rho_c > 0$ generates positive dependence in which responders tend to have longer survival times. For binary--continuous pairs, the attainable range of Pearson correlation is strictly narrower than $[-1, 1]$ \citep{demirtas2011correlation}, so the realized Pearson correlation $\rho_{TX} = \Corr(T_i, X_i)$ is generally smaller than the copula parameter $\rho_c$. In our simulations, the realized $\rho_{TX}$ is estimated directly from the generated data.

Second, the binary safety endpoint is generated conditional on efficacy, using the joint probabilities from Equation~\eqref{eq:pi_all} to obtain $P(Y=1 \mid X=x)$ for $x \in \{0, 1\}$. This preserves both the specified marginal rates $(p, q)$ and the target correlation $\phi$ exactly.

For TTE simulations, patients enroll uniformly over $[0, t_{\text{entry}}]$ with administrative censoring at calendar time $t_{\text{admin}}$, yielding patient-specific follow-up $C_i = t_{\text{admin}} - E_i$ where $E_i$ is the enrollment time. The observed time is $V_i = \min(T_i, C_i)$ with event indicator $\delta_i = \mathbf{1}\{T_i \leq C_i\}$.

\subsubsection{Simulation Scenarios}

For each  scenario, $10^6$ replications were simulated to ensure that Monte Carlo standard errors were negligible relative to the quantities of interest (e.g., around 0.00016--0.0005 for all estimated proportions, such as Type~I error rates and PCS). 

For the sample size calculation validation of response rate (Table~\ref{tab:combined_utility_rose}), we evaluated 48 parameter combinations with key design parameters included: response rates $p \in \{0.3, 0.5\}$, no-AE rates $q \in \{0.5, 0.7\}$, efficacy margins $\delta \in \{0.10, 0.15\}$, safety margin $d = 0.15$, efficacy-safety correlations $\phi \in \{-0.2, 0.0, 0.2\}$, and target PCS levels $\alpha_L = \alpha_H \in \{0.7, 0.8\}$. Utility scores were derived using the margin-based approach of Equation~\eqref{eq:utility_formula}, yielding $u = (1, 0.6, 0.4, 0)$ for $\delta = 0.10$ and $u = (1, 0.5, 0.5, 0)$ for $\delta = 0.15$.

For the validation of bias and Type~I error estimation (Tables~\ref{tab:bias_results}--\ref{tab:type1_results}), we simulated the dose optimization stage under the null hypothesis ($p_L = p_H = p$, $q_L = q_H = 0.8$) with fixed total sample size $n_1 + n_2 = 200$ while varying the Stage~1 sample size $n_1 \in \{40, 60, 80, 100\}$. Response rates $p \in \{0.3, 0.4, 0.5\}$ and efficacy--safety correlations $\phi \in \{0.0, -0.3\}$ were examined, yielding 24 unique parameter combinations. The use of $\phi=-0.3$ is to represent the most clinically likely situation in which better efficacy is associated with worse safety. Type~I error was evaluated using both one-sided Z-tests and exact binomial tests at nominal $\alpha = 0.025$. 

For the TTE confirmatory endpoint analysis (Table~\ref{tab:tte_type1}), the same 24 design configurations were simulated at each of three efficacy--survival correlation levels $\rho_c \in \{0, 0.3, 0.7\}$, where $\rho_c$ is the Gaussian copula parameter. The resulting Pearson correlations between response and survival were approximately $\hat{\rho}_{TX} \approx 0$, $0.22$, and $0.53$ for the three settings, respectively. Survival times followed an exponential distribution with rate $\lambda_0 = 0.1$, subject to uniform accrual over 52 weeks and administrative censoring at 76 weeks. Landmark survival was evaluated at $\tau = 24$ weeks. Type~I error was assessed at one-sided $\alpha = 0.025$ for all TTE tests.

In addition to the three tests for which analytical Type I error expressions are derived (landmark survival Z-test, one-sample exponential test, and two-sample Wald test), we also include the two-sample Cox score test based on the partial likelihood score statistic, which more closely reflects the test statistic used in practice for the confirmatory Cox model. The Cox score test shares the same plugin estimator as the Wald test statistic (Equation~\eqref{eq:type1_cox_main}), since both arise from the same bias propagation chain.

\subsection{Sample Size and PCS Validation}

Table~\ref{tab:combined_utility_rose} presents sample sizes and empirical PCS from both the approximate normal method (Equation~\eqref{eq:sample_size}) and the exact Multinomial method (Equations~\eqref{eq:exact_pcs_L}--\eqref{eq:exact_sample_size}), alongside ROSE design results for comparison.

Both methods consistently achieved the target PCS levels across all 48 scenarios. For the approximate method, empirical PCS values ranged from 0.700 to 0.718 when targeting 0.70, and from 0.800 to 0.806 when targeting 0.80. The exact method produced slightly higher sample sizes and larger PCS values as compared with the approximate method.

Larger efficacy margins ($\delta = 0.15$ vs. $0.10$) substantially reduced the required sample sizes. For example, with $p = 0.3$, $q = 0.5$, $\phi = 0$, and PCS = 0.8, the exact method yielded $n=46$ for $\delta = 0.10$ versus $n = 33$ for $\delta = 0.15$. Positive correlation between efficacy and safety increased sample size, because when good outcomes (response and no-AE) cluster together, utility scores have less discriminating power between doses, requiring more patients to reliably differentiate them. For example, comparing $\phi = -0.2$ to $\phi = 0.2$ for $p = 0.3$, $\delta=0.10$, PCS = 0.8: exact $n$ increased from 38 to 54. In contrast, the no--AE rate $q$ had a modest effect on sample sizes, with the maximum occurring at $q=0.5$ where the Bernoulli variance is maximized. This is because the utility difference is determined by the margins $\delta$ and $d$, and $q$ influences only the variance of the utility distribution through $q(1-q)$ terms.

The utility-based design required fewer patients than the ROSE design, which relies solely on efficacy. For instance, with $p = 0.3$, $q=0.5$, $\delta = 0.10$, $\alpha = 0.8$: the utility exact method required $n = 38$--$54$ per arm depending on $\phi$, while ROSE exact method required $n = 122$ per arm---a reduction of 56--69\%. This efficiency arises because the composite utility score has lower variance than the response rate alone when safety information contributes meaningfully to dose selection.

\subsection{Bias Validation}

Table~\ref{tab:bias_results} compares observed selection-induced bias from simulations against the analytical estimates from Equation~\eqref{eq:bias_final} (utility-based plugin estimator) and Equation~\eqref{eq:bias_max} (response-only maximum bound).

The analytical bias formulas matched simulation results closely across all 24 null hypothesis scenarios (mean absolute error $< 0.00013$, maximum absolute error $< 0.00025$). For example, with $p = 0.4$, $\phi = 0$, $n_1 = 60$: observed bias was 0.01051 versus estimated 0.01056 and maximum bound 0.01077. The maximum bound (Equation~\eqref{eq:bias_max}) consistently exceeded observed values in all 24 scenarios.

Under the null hypothesis, dose selection is driven primarily by random variation in the data rather than by the within-patient correlation between efficacy and safety. Consequently, the correlation parameter $\phi$ had negligible impact on the magnitude of bias, as expected from the theoretical framework. For instance, with $p = 0.3$, $n_1 = 60$, observed biases were 0.00978 ($\phi = 0$) and 0.00975 ($\phi = -0.3$), demonstrating nearly identical values. This confirms that the bias formula correctly captures the selection mechanism independent of the efficacy--safety correlation structure.

The fixed total sample size ($n_1 + n_2 = 200$) isolated the effect of Stage~1 proportion on the combined bias. Comparing across $n_1 \in \{40, 60, 80, 100\}$, the combined bias increased with $n_1$ due to the weighting factor $n_1/(n_1 + n_2)$ in Equation~\eqref{eq:combined_bias}. For $p = 0.4$, $\phi = 0$: combined bias increased from 0.00859 ($n_1 = 40$) to 0.01355 ($n_1 = 100$). This pattern was consistent across all response rates and correlation values.

\subsection{Type I Error Assessment}

Table~\ref{tab:type1_results} presents observed and estimated Type~I error rates for both Z-tests and exact binomial tests. Selection-induced bias resulted in Type~I error inflation above the nominal $\alpha = 0.025$ level for both test procedures. Both plugin estimators are conservative, overestimating the true Type~I error rate.

\subsubsection{Z-test Performance}

The one-sided Z-test showed a mean observed Type~I error of 0.0464 across all 24 scenarios (range: 0.0407--0.0509), corresponding to an inflation factor of $1.86\times$ relative to the nominal $\alpha = 0.025$. The plugin estimator for the Z-test predicted mean Type~I error of 0.0512, overestimating the observed rate by approximately 0.0049 on average. The plugin was conservative (Est $>$ Observed) in 23 of 24 scenarios (96\%). This conservative property is desirable for regulatory planning, as it provides a reliable upper bound for the anticipated inflation. Consistent with the bias results, Type~I error inflation increased monotonically with $n_1$ (e.g., from 0.0407 at $n_1 = 40$ to 0.0478 at $n_1 = 100$ for $p = 0.4$, $\phi = 0$) and was negligibly sensitive to $\phi$.

\subsubsection{Binomial Test Performance}

The exact binomial test exhibited substantially lower Type~I error inflation, with mean observed error of 0.0334 (range: 0.0294--0.0365), corresponding to an inflation factor of $1.34\times$. The plugin estimator predicted mean Type~I error of 0.0416, overestimating the observed rate by approximately 0.0081 on average. The plugin was conservative in all 24 scenarios (100\%). The same monotonic dependence on $n_1$ was observed (e.g., from 0.0295 at $n_1 = 40$ to 0.0344 at $n_1 = 100$ for $p = 0.4$, $\phi = 0$).

The greater conservatism of the binomial plugin estimator (overestimation of 0.0081 vs.\ 0.0049 for the Z-test) arises because the binomial test's discrete rejection region creates a gap between the nominal and effective significance levels, which the continuous plugin formula does not account for. This means the plugin substantially overestimates the actual inflation, providing an even safer bound for planning.

\subsubsection{TTE Confirmatory Endpoint Performance}

Table~\ref{tab:tte_type1} presents Type~I error rate results by the efficacy--survival correlation strength for four confirmatory test procedures applied to pooled data (Stage~1 selected arm + Stage~2) under the null hypothesis with nominal one-sided $\alpha = 0.025$: (i) a landmark survival Z-test using the binary indicator $S_i(\tau) = \mathbf{1}\{T_i > 24\text{ weeks}\}$; (ii) a one-sample exponential test comparing the selected arm's survival to a historical control; (iii) a two-sample log-rank test comparing the selected arm to a concurrent control; and (iv) a two-sample Cox score test based on the partial likelihood score statistic. As with the binary endpoints, the efficacy--safety correlation $\phi$ had negligible impact on TTE Type~I error rates across all four tests and all three $\rho_c$ levels (differences $< 0.0004$), and larger $n_1$ produced monotonically higher inflation; results below are reported for $\phi = 0$.

\paragraph{Landmark Survival Z-test:} The landmark Z-test maintained near-nominal size when response and survival were independent ($\rho_c=0$), with a mean observed error of 0.0245 ($0.98 \times$ the nominal $\alpha=0.025$). As the response--survival correlation increased, mild inflation emerged, with the error rate rising to 0.0271 ($1.08 \times$) at $\rho_c=0.3$ and 0.0309 ($1.24 \times$) at $\rho_c=0.7$.
The plugin estimator for this test was uniformly conservative across all scenarios and all three correlation levels, overestimating the true Type~I error in 100\% of the simulation cases. For example at $\rho_c = 0.7$, the mean overestimation is 0.0011; at $\rho_c = 0.3$, it is 0.0008; at $\rho_c = 0$, it is 0.0006. This combination of near-nominal performance and a conservative estimator makes the landmark survival rate an attractive endpoint when pooling data from a dose optimization study with a subsequent registration trial.

\paragraph{One-sample Exponential Test:} The one-sample exponential test was consistently the most conservative. Even under strong efficacy--survival correlation ($\rho_c=0.7$), its mean observed error rate was only 0.0296 ($1.18 \times$), and it remained at or below the nominal $\alpha$ level for $\rho_c \le 0.3$. On the other hand, the plugin estimator of the one-sample exponential test showed the highest predicted error rate with the largest overestimation margin over the observed error rate among all four TTE endpoints (mean overestimation 0.0064 at $\rho_c = 0.7$, 0.0044 at $\rho_c = 0.3$, 0.0040 at $\rho_c = 0$). These properties trace to the convexity of $g(x)=1/x$ underlying $\hat{\lambda} = 1/\bar{T}$. By Jensen's inequality, $E[1/\bar{T}] > 1/E[\bar{T}]$, where the expectation is taken over the sampling distribution of $\bar{T}$ across repeated simulation trials. Thus $\hat{\lambda}$ is positively biased for $\lambda_0$. This upward shift in $\hat{\lambda}$ translates to a positive mean shift in $Z_{\text{Exp}}$, which reduces the probability mass in the rejection tail and causes the test to run below its nominal level even in the absence of selection bias ($\rho_c = 0$). The same convexity causes the first-order delta method in the plugin formula to overestimate the true bias magnitude. Administrative censoring widens this gap by truncating the long survival times where selection effects are most pronounced, further reducing the observed bias relative to what the plugin predicts. This property of conservatism makes the one-sample exponential test a robust and safe choice from a regulatory perspective.

\paragraph{Two-sample Log-Rank Test:} At $\rho_c=0$, the two-sample log-rank test closely tracked the nominal $\alpha$ level (mean 0.0254, or $1.02\times$). However, it exhibited higher inflation than the landmark and one-sample tests, rising to 0.0321 ($1.28 \times$) at $\rho_c=0.7$. This is because the two-sample log-rank test uses the full distribution of event times rather than a dichotomized landmark or a convexity-prone hazard estimate, so it has no structural mechanism that suppresses rejection. The same sensitivity that makes it the most powerful test for detecting true hazard differences also makes it the most responsive to selection-induced shifts under the null \citep{harrington1982class}. The plug-in estimator of the two-sample log-rank test closely tracks the observed value, and the direction of approximation error depends on the efficacy--survival correlation (83\% conservative at $\rho_c = 0.7$, 8\% at $\rho_c = 0.3$, 0\% at $\rho_c = 0$).

\paragraph{Two-sample Cox Score Test:} The Cox score test, based on the partial likelihood score statistic, showed the highest inflation among all four TTE tests: mean observed Type~I error of 0.0325 ($1.30\times$) at $\rho_c = 0.7$, 0.0284 ($1.14\times$) at $\rho_c = 0.3$, and 0.0256 ($1.02\times$) at $\rho_c = 0$. These rates are slightly higher than those of the log-rank test at each correlation level, consistent with the score test's different weighting of events across time. The plugin estimator is approximately unbiased but slightly anti-conservative, with mean underestimation of 0.0001 at $\rho_c = 0.7$ (conservative in 33\% of scenarios), 0.0005 at $\rho_c = 0.3$ (0\%), and 0.0005 at $\rho_c = 0$ (0\%). In practice, the maximum underestimation is less than 0.001 across all scenarios.

\section{Discussion}
\label{sec:discussion}

This paper develops a Frequentist framework for one-stage two-sample dose optimization studies, directly addressing the need to balance efficacy and safety as outlined by the FDA's \textit{Project Optimus} initiative. Such studies typically serve as the final step before Phase~3 registration trials, where the primary objective is to identify the superior dose from two candidates with acceptable efficacy-safety profiles. The one-stage design is well-suited to this setting: it provides a transparent sample size justification, avoids the complexity of adaptive stopping rules, and yields a straightforward decision at study completion.

The three contributions of our work reinforce each other in practice: utility scores calibrated to clinical trade-offs naturally define the scenarios used for sample sizing, and the bias formulas use those same margins to bound downstream error inflation. The bias characterization, while presented in the context of a hypothetical two-stage development pathway, is intended to inform confirmatory trial planning, not to propose an adaptive multi-stage design. Extensive simulation studies validate the framework: the sample size methods reliably achieved their target PCS across 48 diverse scenarios (Table~\ref{tab:combined_utility_rose}); the analytical bias expressions demonstrated excellent concordance with empirical estimates across 24 null hypothesis scenarios, with a maximum absolute error of less than 0.00025 (Table~\ref{tab:bias_results}); and the utility-based design achieved the same selection accuracy as the efficacy-only ROSE design with a meaningful reduction in the required sample size in all scenarios. A key finding is the robustness of the bias formula to the efficacy--safety correlation $\phi$, a significant practical advantage since $\phi$ is often difficult to estimate precisely from historical data.

The practical application of our framework is straightforward. Study teams can anchor the design in clinical reality by first defining the trade-off ratio $r = \delta/d$ to generate utility scores (Equation~\eqref{eq:utility_formula}) and then using the closed-form equations to calculate the required sample size. The bias formulas (Equations~\eqref{eq:bias_final}--\eqref{eq:bias_max}) then allow for prospective planning for a confirmatory stage. Because Type~I error inflation increases monotonically with the Stage~1 proportion $n_1/(n_1 + n_2)$ in the combined analysis (Tables~\ref{tab:type1_results}--\ref{tab:tte_type1}), the design involves a fundamental trade-off: a larger $n_1$ improves PCS in the dose optimization stage but increases the potential for Type~I error inflation when Stage~1 data are pooled with the subsequent confirmatory trial. Sponsors can use the plugin estimators to quantify this trade-off prospectively and choose $n_1$ to balance selection accuracy against error-rate control.

For the confirmatory binary endpoint, the choice between the Z-test and the exact binomial test involves a trade-off between inflation magnitude and conservatism. Both plugin estimators are conservative: the Z-test overestimates by approximately 0.005 (96\% conservative) and the binomial test overestimates by approximately 0.008 (100\% conservative). The binomial test offers substantially lower Type~I error inflation ($1.34\times$ vs $1.86\times$), while the Z-test provides a tighter (less conservative) plugin bound. For regulatory submissions, the binomial test is the recommended choice: it has lower inflation, a uniformly conservative plugin estimator, and the additional overestimation provides a built-in safety margin. The Z-test remains useful when a continuous approximation is preferred or when the conservative margin of the binomial plugin is unnecessarily wide.

For TTE confirmatory endpoints, the simulation results (Tables~\ref{tab:tte_type1} and \ref{tab:tte_summary}) demonstrate that the magnitude of Type~I error inflation is governed by the strength of the correlation between binary efficacy and survival, rather than by the specific test procedure. When response and survival are independent ($\rho_c = 0$), selection bias does not propagate to TTE endpoints and all four tests maintain their nominal size. At moderate correlation ($\rho_c = 0.3$, $\hat{\rho}_{TX} \approx 0.22$), all TTE tests remain near-nominal ($0.99$--$1.14\times$). Even at strong correlation ($\rho_c = 0.7$, $\hat{\rho}_{TX} \approx 0.53$), inflation factors range from only $1.18\times$ to $1.30\times$---substantially attenuated relative to the $1.85\times$ observed for the binary Z-test. This attenuation reflects the indirect pathway through which selection bias reaches TTE endpoints: bias must propagate from the utility score through the binary efficacy endpoint and then through the response--survival correlation to the survival estimate, with each step introducing a multiplicative attenuation factor.

Among the four TTE tests, the one-arm tests (landmark Z-test and one-sample exponential) have uniformly conservative plugin estimators (100\% across all $\rho_c$ levels), providing reliable upper bounds for regulatory planning. The two-sample tests (log-rank and Cox score) share the same plugin estimator derived from the Wald test structure. This plugin estimation tracks the observed value closely, a property that follows directly from the Cox model framework: under the null hypothesis, $\mathbb{E}[\hat{\beta}] = \Bias(\hat{\beta})_{\text{comb}}$ exactly (Equation~\eqref{eq:beta_bias}), and since Fisher information under equal allocation is determined by trial design and unaffected by dose selection, Stage~1 data can exactly recover the mean shift in $Z_{\text{Cox}}$ at the end of Stage~2 via Equation~\eqref{eq:bias_z_cox}. The plugin is conservative for the log-rank test (83\% at $\rho_c = 0.7$) and modestly anti-conservative for the Cox score test (33\% at $\rho_c = 0.7$), with maximum underestimation less than 0.001 in all cases. The anti-conservatism of the Cox score test plugin arises because the score test's data-adaptive variance estimation from observed risk sets captures marginally more variability than the fixed-information approximation underlying the plugin formula. In practice, study teams can estimate $\rho_{TX}$ from historical data---for example, from landmark analyses of prior trials relating response to progression-free survival---and use the plugin formulas to prospectively bound the inflation for their specific setting. When $\hat{\rho}_{TX} \leq 0.3$, TTE confirmatory analyses face minimal inflation even without bias correction; at stronger correlation, the modest inflation ($\leq 1.30\times$) may warrant a conservative alpha adjustment in registration trials.

While the proposed framework is robust, it has limitations that highlight areas for future research. The utility score methodology is currently specified for binary efficacy and safety endpoints; extending it to handle ordinal, continuous, or time-to-event outcomes directly would broaden its applicability. The bias expressions are derived under the null hypothesis, providing a conservative estimate that may overstate inflation when a true difference exists between doses. Our TTE analysis assumes that selection bias propagates to survival exclusively through the binary efficacy endpoint, whereas in practice safety endpoints may also correlate with survival, providing an additional pathway through the utility score; modeling this joint dependence structure is a natural direction for future work. The framework is also developed for the two-dose comparison setting, and extension to three or more candidate doses would require generalizing the selection mechanism and the corresponding bias characterization. On the simulation side, our TTE scenarios assumed exponential survival with administrative censoring, but real trials may feature non-proportional hazards, informative censoring, or non-exponential baseline hazards, and the robustness of the bias propagation chain under these departures warrants further investigation. Finally, while our one-stage dose optimization study is designed for late-stage dose selection before a registration trial, a natural extension is to develop a multi-stage adaptive framework for earlier clinical phases, where greater uncertainty can be managed with interim analyses. Building on established methods \citep{bauer1994evaluation, OBrienFleming1979}, such a design could reduce the expected sample size by allowing early stopping, while preserving the target PCS and controlling the family-wise error rate.

\section*{Software Availability}
The R package \texttt{DoseOptDesign} implementing all methods described in this paper---including sample size calculation (approximate and exact), utility score specification, bias estimation, and plugin Type~I error prediction---is available at \url{https://github.com/gux9/DoseOptDesign}. 

\section*{Acknowledgment}
The authors used large language models, Claude Opus 4.6 by Anthropic, to assist with revising, rephrasing for clarity, and grammar checking. All generated content was carefully reviewed, substantially revised, and verified for scientific accuracy and integrity by the human authors, who take full responsibility for the final manuscript.

\bibliographystyle{apalike}

\clearpage
\begingroup
\centering
\small
\setlength\tabcolsep{2pt}
\setlength\LTpost{0pt}
\setlength\LTpre{0pt}

\begin{xltabular}{\textwidth}{cccccc|ccc|ccc|ccc|ccc}
\caption{Sample Size Comparison: Approximate vs Exact Methods for Utility-Based Design (with ROSE Design results for comparison)} 
\label{tab:combined_utility_rose} \\
\hline\hline
\multicolumn{6}{c|}{\textbf{Design Parameters}} & 
\multicolumn{3}{c|}{\textbf{Util.\ Approx.}} & 
\multicolumn{3}{c|}{\textbf{Util.\ Exact}} & 
\multicolumn{3}{c|}{\textbf{ROSE Approx.}} & 
\multicolumn{3}{c}{\textbf{ROSE Exact}} \\ 
\cline{1-6}\cline{7-9}\cline{10-12}\cline{13-15}\cline{16-18}
$\alpha_L/\alpha_H$ & $p$ & $q$ & $\delta$ & $d$ & $\phi$ & 
$n$ & PCS$_L$ & PCS$_H$ & 
$n$ & PCS$_L$ & PCS$_H$ & 
$n$ & PCS$_L$ & PCS$_H$ & 
$n$ & PCS$_L$ & PCS$_H$ \\
\hline
\endfirsthead
\caption[]{Sample Size Comparison -- continued} \\
\hline\hline
\multicolumn{6}{c|}{\textbf{Design Parameters}} & 
\multicolumn{3}{c|}{\textbf{Util.\ Approx.}} & 
\multicolumn{3}{c|}{\textbf{Util.\ Exact}} & 
\multicolumn{3}{c|}{\textbf{ROSE Approx.}} & 
\multicolumn{3}{c}{\textbf{ROSE Exact}} \\ 
\cline{1-6}\cline{7-9}\cline{10-12}\cline{13-15}\cline{16-18}
$\alpha_L/\alpha_H$ & $p$ & $q$ & $\delta$ & $d$ & $\phi$ & 
$n$ & PCS$_L$ & PCS$_H$ & 
$n$ & PCS$_L$ & PCS$_H$ & 
$n$ & PCS$_L$ & PCS$_H$ & 
$n$ & PCS$_L$ & PCS$_H$ \\
\hline
\endhead

0.7 & 0.3 & 0.5 & 0.10 & 0.15 & -0.2 & 14 & 0.705 & 0.700 & 17 & 0.736 & 0.705 & 44 & 0.702 & 0.700 & 47 & 0.713 & 0.702 \\
0.7 & 0.3 & 0.5 & 0.10 & 0.15 & 0.0 & 17 & 0.702 & 0.700 & 20 & 0.729 & 0.703 & 44 & 0.702 & 0.700 & 47 & 0.713 & 0.702 \\
0.7 & 0.3 & 0.5 & 0.10 & 0.15 & 0.2 & 21 & 0.708 & 0.700 & 23 & 0.724 & 0.702 & 44 & 0.702 & 0.700 & 47 & 0.713 & 0.702 \\
0.7 & 0.3 & 0.7 & 0.10 & 0.15 & -0.2 & 14 & 0.711 & 0.700 & 16 & 0.732 & 0.703 & 44 & 0.702 & 0.700 & 47 & 0.713 & 0.702 \\
0.7 & 0.3 & 0.7 & 0.10 & 0.15 & 0.0 & 17 & 0.708 & 0.700 & 19 & 0.726 & 0.703 & 44 & 0.702 & 0.700 & 47 & 0.713 & 0.702 \\
0.7 & 0.3 & 0.7 & 0.10 & 0.15 & 0.2 & 20 & 0.707 & 0.700 & 22 & 0.721 & 0.703 & 44 & 0.702 & 0.700 & 47 & 0.713 & 0.702 \\
\hline
0.7 & 0.3 & 0.5 & 0.15 & 0.15 & -0.2 & 9 & 0.710 & 0.700 & 15 & 0.800 & 0.709 & 19 & 0.705 & 0.700 & 19 & 0.703 & 0.704 \\
0.7 & 0.3 & 0.5 & 0.15 & 0.15 & 0.0 & 11 & 0.707 & 0.700 & 17 & 0.783 & 0.707 & 19 & 0.705 & 0.700 & 19 & 0.703 & 0.704 \\
0.7 & 0.3 & 0.5 & 0.15 & 0.15 & 0.2 & 13 & 0.704 & 0.700 & 19 & 0.771 & 0.705 & 19 & 0.705 & 0.700 & 19 & 0.703 & 0.704 \\
0.7 & 0.3 & 0.7 & 0.15 & 0.15 & -0.2 & 8 & 0.700 & 0.700 & 14 & 0.797 & 0.709 & 19 & 0.705 & 0.700 & 19 & 0.703 & 0.704 \\
0.7 & 0.3 & 0.7 & 0.15 & 0.15 & 0.0 & 10 & 0.700 & 0.700 & 16 & 0.781 & 0.708 & 19 & 0.705 & 0.700 & 19 & 0.703 & 0.704 \\
0.7 & 0.3 & 0.7 & 0.15 & 0.15 & 0.2 & 12 & 0.701 & 0.700 & 18 & 0.769 & 0.707 & 19 & 0.705 & 0.700 & 19 & 0.703 & 0.704 \\
\hline
0.7 & 0.5 & 0.5 & 0.10 & 0.15 & -0.2 & 16 & 0.700 & 0.700 & 20 & 0.737 & 0.705 & 55 & 0.702 & 0.700 & 65 & 0.730 & 0.703 \\
0.7 & 0.5 & 0.5 & 0.10 & 0.15 & 0.0 & 20 & 0.704 & 0.700 & 23 & 0.729 & 0.701 & 55 & 0.702 & 0.700 & 65 & 0.730 & 0.703 \\
0.7 & 0.5 & 0.5 & 0.10 & 0.15 & 0.2 & 24 & 0.706 & 0.700 & 27 & 0.726 & 0.702 & 55 & 0.702 & 0.700 & 65 & 0.730 & 0.703 \\
0.7 & 0.5 & 0.7 & 0.10 & 0.15 & -0.2 & 16 & 0.705 & 0.700 & 19 & 0.734 & 0.703 & 55 & 0.702 & 0.700 & 65 & 0.730 & 0.703 \\
0.7 & 0.5 & 0.7 & 0.10 & 0.15 & 0.0 & 19 & 0.700 & 0.700 & 22 & 0.726 & 0.701 & 55 & 0.702 & 0.700 & 65 & 0.730 & 0.703 \\
0.7 & 0.5 & 0.7 & 0.10 & 0.15 & 0.2 & 23 & 0.705 & 0.700 & 26 & 0.724 & 0.704 & 55 & 0.702 & 0.700 & 65 & 0.730 & 0.703 \\
\hline
0.7 & 0.5 & 0.5 & 0.15 & 0.15 & -0.2 & 10 & 0.708 & 0.700 & 16 & 0.794 & 0.705 & 24 & 0.701 & 0.700 & 30 & 0.741 & 0.703 \\
0.7 & 0.5 & 0.5 & 0.15 & 0.15 & 0.0 & 12 & 0.701 & 0.700 & 18 & 0.777 & 0.701 & 24 & 0.701 & 0.700 & 30 & 0.741 & 0.703 \\
0.7 & 0.5 & 0.5 & 0.15 & 0.15 & 0.2 & 15 & 0.708 & 0.700 & 21 & 0.769 & 0.704 & 24 & 0.701 & 0.700 & 30 & 0.741 & 0.703 \\
0.7 & 0.5 & 0.7 & 0.15 & 0.15 & -0.2 & 10 & 0.718 & 0.700 & 15 & 0.792 & 0.704 & 24 & 0.701 & 0.700 & 30 & 0.741 & 0.703 \\
0.7 & 0.5 & 0.7 & 0.15 & 0.15 & 0.0 & 12 & 0.710 & 0.700 & 17 & 0.775 & 0.701 & 24 & 0.701 & 0.700 & 30 & 0.741 & 0.703 \\
0.7 & 0.5 & 0.7 & 0.15 & 0.15 & 0.2 & 14 & 0.705 & 0.700 & 20 & 0.767 & 0.706 & 24 & 0.701 & 0.700 & 30 & 0.741 & 0.703 \\
\hline
0.8 & 0.3 & 0.5 & 0.10 & 0.15 & -0.2 & 36 & 0.806 & 0.800 & 38 & 0.815 & 0.803 & 112 & 0.800 & 0.800 & 122 & 0.818 & 0.802 \\
0.8 & 0.3 & 0.5 & 0.10 & 0.15 & 0.0 & 44 & 0.803 & 0.800 & 46 & 0.812 & 0.802 & 112 & 0.800 & 0.800 & 122 & 0.818 & 0.802 \\
0.8 & 0.3 & 0.5 & 0.10 & 0.15 & 0.2 & 52 & 0.802 & 0.800 & 54 & 0.809 & 0.801 & 112 & 0.800 & 0.800 & 122 & 0.818 & 0.802 \\
0.8 & 0.3 & 0.7 & 0.10 & 0.15 & -0.2 & 34 & 0.801 & 0.800 & 36 & 0.811 & 0.803 & 112 & 0.800 & 0.800 & 122 & 0.818 & 0.802 \\
0.8 & 0.3 & 0.7 & 0.10 & 0.15 & 0.0 & 42 & 0.801 & 0.800 & 43 & 0.806 & 0.801 & 112 & 0.800 & 0.800 & 122 & 0.818 & 0.802 \\
0.8 & 0.3 & 0.7 & 0.10 & 0.15 & 0.2 & 50 & 0.802 & 0.800 & 51 & 0.804 & 0.802 & 112 & 0.800 & 0.800 & 122 & 0.818 & 0.802 \\
\hline
0.8 & 0.3 & 0.5 & 0.15 & 0.15 & -0.2 & 22 & 0.801 & 0.800 & 28 & 0.853 & 0.803 & 48 & 0.802 & 0.800 & 54 & 0.828 & 0.801 \\
0.8 & 0.3 & 0.5 & 0.15 & 0.15 & 0.0 & 28 & 0.806 & 0.800 & 33 & 0.842 & 0.801 & 48 & 0.802 & 0.800 & 54 & 0.828 & 0.801 \\
0.8 & 0.3 & 0.5 & 0.15 & 0.15 & 0.2 & 33 & 0.802 & 0.800 & 38 & 0.834 & 0.800 & 48 & 0.802 & 0.800 & 54 & 0.828 & 0.801 \\
0.8 & 0.3 & 0.7 & 0.15 & 0.15 & -0.2 & 21 & 0.804 & 0.800 & 26 & 0.849 & 0.804 & 48 & 0.802 & 0.800 & 54 & 0.828 & 0.801 \\
0.8 & 0.3 & 0.7 & 0.15 & 0.15 & 0.0 & 26 & 0.803 & 0.800 & 31 & 0.838 & 0.804 & 48 & 0.802 & 0.800 & 54 & 0.828 & 0.801 \\
0.8 & 0.3 & 0.7 & 0.15 & 0.15 & 0.2 & 31 & 0.801 & 0.800 & 35 & 0.828 & 0.801 & 48 & 0.802 & 0.800 & 54 & 0.828 & 0.801 \\
\hline
0.8 & 0.5 & 0.5 & 0.10 & 0.15 & -0.2 & 42 & 0.805 & 0.800 & 45 & 0.819 & 0.802 & 141 & 0.801 & 0.800 & 147 & 0.809 & 0.802 \\
0.8 & 0.5 & 0.5 & 0.10 & 0.15 & 0.0 & 51 & 0.803 & 0.800 & 54 & 0.815 & 0.801 & 141 & 0.801 & 0.800 & 147 & 0.809 & 0.802 \\
0.8 & 0.5 & 0.5 & 0.10 & 0.15 & 0.2 & 60 & 0.801 & 0.800 & 64 & 0.814 & 0.802 & 141 & 0.801 & 0.800 & 147 & 0.809 & 0.802 \\
0.8 & 0.5 & 0.7 & 0.10 & 0.15 & -0.2 & 41 & 0.805 & 0.800 & 43 & 0.815 & 0.801 & 141 & 0.801 & 0.800 & 147 & 0.809 & 0.802 \\
0.8 & 0.5 & 0.7 & 0.10 & 0.15 & 0.0 & 49 & 0.801 & 0.800 & 52 & 0.812 & 0.802 & 141 & 0.801 & 0.800 & 147 & 0.809 & 0.802 \\
0.8 & 0.5 & 0.7 & 0.10 & 0.15 & 0.2 & 58 & 0.801 & 0.800 & 60 & 0.809 & 0.800 & 141 & 0.801 & 0.800 & 147 & 0.809 & 0.802 \\
\hline
0.8 & 0.5 & 0.5 & 0.15 & 0.15 & -0.2 & 25 & 0.804 & 0.800 & 31 & 0.852 & 0.801 & 62 & 0.802 & 0.800 & 69 & 0.825 & 0.801 \\
0.8 & 0.5 & 0.5 & 0.15 & 0.15 & 0.0 & 31 & 0.802 & 0.800 & 38 & 0.845 & 0.803 & 62 & 0.802 & 0.800 & 69 & 0.825 & 0.801 \\
0.8 & 0.5 & 0.5 & 0.15 & 0.15 & 0.2 & 37 & 0.800 & 0.800 & 44 & 0.838 & 0.802 & 62 & 0.802 & 0.800 & 69 & 0.825 & 0.801 \\
0.8 & 0.5 & 0.7 & 0.15 & 0.15 & -0.2 & 24 & 0.806 & 0.800 & 29 & 0.848 & 0.801 & 62 & 0.802 & 0.800 & 69 & 0.825 & 0.801 \\
0.8 & 0.5 & 0.7 & 0.15 & 0.15 & 0.0 & 30 & 0.806 & 0.800 & 35 & 0.840 & 0.802 & 62 & 0.802 & 0.800 & 69 & 0.825 & 0.801 \\
0.8 & 0.5 & 0.7 & 0.15 & 0.15 & 0.2 & 36 & 0.806 & 0.800 & 41 & 0.834 & 0.803 & 62 & 0.802 & 0.800 & 69 & 0.825 & 0.801 \\

\hline\hline
\end{xltabular}
\endgroup

\vspace{6pt}
\noindent\textbf{Table Notes:}
\textbf{Common definitions}: ($p$, $q$): response rate and no-AE rate for Dose $L$ in Scenario $\mathcal{S}_L$ and Dose $H$ in Scenario $\mathcal{S}_H$ (see Table~\ref{tab:selection_criteria01}); $\delta$ = efficacy margin; 
PCS$_L$ = probability of correct selection when doses equal (Scenario L); 
PCS$_H$ = probability of correct selection when high dose superior (Scenario H); $\alpha_L = \alpha_H = 0.7$ or $0.8$.
\textbf{Utility-Based Design}: $d$ = safety margin, $\phi$ = correlation, $r = \delta/d$. 
Utility scores: $u = (1, 1/(1+r), r/(1+r), 0)$.
\textbf{ROSE Design}: No explicit safety parameters ($\phi$, $q$, $d$ not applicable). 
Utility scores: $u = (1, 1, 0, 0)$. ROSE results are shown in last six columns for comparable $(p, \delta)$ settings.

\clearpage

\begin{table}[htbp]
\centering
\small
\begin{threeparttable}
\caption{Selection-Induced Bias Under Null Hypothesis}
\label{tab:bias_results}
\setlength\tabcolsep{4pt}
\begin{tabular}{cccc|ccc|c}
\hline\hline
\multicolumn{4}{c|}{\textbf{Design Parameters}} &
\multicolumn{3}{c|}{\textbf{Combined Bias}} &
\\
\hline
$p$ & $\phi$ & $n_1$ & $n_1+n_2$ &
Observed & Est & Est$_{\text{max}}$ & $\hat{\phi}$ \\
\hline
0.3 & 0.0  &  40 & 200 & 0.00803 & 0.00820 & 0.00837 & 0.000 \\
0.3 & 0.0  &  60 & 200 & 0.00978 & 0.01001 & 0.01023 & 0.000 \\
0.3 & 0.0  &  80 & 200 & 0.01132 & 0.01154 & 0.01179 & 0.000 \\
0.3 & 0.0  & 100 & 200 & 0.01265 & 0.01288 & 0.01316 & 0.000 \\
0.4 & 0.0  &  40 & 200 & 0.00859 & 0.00862 & 0.00878 & 0.000 \\
0.4 & 0.0  &  60 & 200 & 0.01051 & 0.01056 & 0.01077 & 0.000 \\
0.4 & 0.0  &  80 & 200 & 0.01212 & 0.01220 & 0.01244 & 0.000 \\
0.4 & 0.0  & 100 & 200 & 0.01355 & 0.01364 & 0.01390 & 0.000 \\
0.5 & 0.0  &  40 & 200 & 0.00875 & 0.00864 & 0.00881 & 0.000 \\
0.5 & 0.0  &  60 & 200 & 0.01072 & 0.01063 & 0.01083 & 0.000 \\
0.5 & 0.0  &  80 & 200 & 0.01238 & 0.01230 & 0.01254 & 0.000 \\
0.5 & 0.0  & 100 & 200 & 0.01383 & 0.01377 & 0.01403 & 0.000 \\
\hline
0.3 & $-$0.3 &  40 & 200 & 0.00800 & 0.00819 & 0.00837 & $-$0.300 \\
0.3 & $-$0.3 &  60 & 200 & 0.00975 & 0.01000 & 0.01023 & $-$0.300 \\
0.3 & $-$0.3 &  80 & 200 & 0.01129 & 0.01152 & 0.01179 & $-$0.300 \\
0.3 & $-$0.3 & 100 & 200 & 0.01264 & 0.01287 & 0.01316 & $-$0.300 \\
0.4 & $-$0.3 &  40 & 200 & 0.00854 & 0.00861 & 0.00878 & $-$0.299 \\
0.4 & $-$0.3 &  60 & 200 & 0.01044 & 0.01055 & 0.01077 & $-$0.298 \\
0.4 & $-$0.3 &  80 & 200 & 0.01210 & 0.01219 & 0.01244 & $-$0.298 \\
0.4 & $-$0.3 & 100 & 200 & 0.01355 & 0.01362 & 0.01390 & $-$0.298 \\
0.5 & $-$0.3 &  40 & 200 & 0.00873 & 0.00863 & 0.00881 & $-$0.297 \\
0.5 & $-$0.3 &  60 & 200 & 0.01070 & 0.01062 & 0.01083 & $-$0.297 \\
0.5 & $-$0.3 &  80 & 200 & 0.01236 & 0.01229 & 0.01254 & $-$0.297 \\
0.5 & $-$0.3 & 100 & 200 & 0.01386 & 0.01375 & 0.01403 & $-$0.296 \\
\hline\hline
\end{tabular}
\begin{tablenotes}[flushleft]
\small
\item \textbf{Simulation Setting}: Null hypothesis $p_L = p_H = p$, $q_L = q_H = 0.8$, utility weights $(u_1,u_2,u_3,u_4)=(1,0.8,0.2,0)$, $n_1 + n_2 = 200$.
Results averaged over $10^6$ replications.
\item \textbf{Columns}: Observed = empirical bias from simulations;
Est = plugin estimate using Equation~\eqref{eq:bias_final};
Est$_{\text{max}}$ = conservative upper bound using Equation~\eqref{eq:bias_max};
$\hat{\phi}$ = estimated efficacy--safety correlation from Stage~1 data.
\end{tablenotes}
\end{threeparttable}
\end{table}

\clearpage

\begin{table}[htbp]
\centering
\small
\begin{threeparttable}
\caption{Type I Error Rate: Z-test vs Binomial Test Under Null Hypothesis}
\label{tab:type1_results}
\setlength\tabcolsep{3pt}
\begin{tabular}{cccc|ccc|ccc}
\hline\hline
\multicolumn{4}{c|}{\textbf{Design Parameters}} &
\multicolumn{3}{c|}{\textbf{Z-test}} &
\multicolumn{3}{c}{\textbf{Binomial Test}} \\
\hline
$p$ & $\phi$ & $n_1$ & $n_1+n_2$ &
Observed & Est & Est$_{\text{max}}$ &
Observed & Est & Est$_{\text{max}}$ \\
\hline
0.3 & 0.0  &  40 & 200 & 0.0439 & 0.0439 & 0.0444 & 0.0314 & 0.0362 & 0.0366 \\
0.3 & 0.0  &  60 & 200 & 0.0467 & 0.0494 & 0.0501 & 0.0333 & 0.0409 & 0.0415 \\
0.3 & 0.0  &  80 & 200 & 0.0490 & 0.0544 & 0.0552 & 0.0350 & 0.0453 & 0.0460 \\
0.3 & 0.0  & 100 & 200 & 0.0509 & 0.0591 & 0.0601 & 0.0365 & 0.0494 & 0.0503 \\
0.4 & 0.0  &  40 & 200 & 0.0407 & 0.0435 & 0.0440 & 0.0295 & 0.0337 & 0.0340 \\
0.4 & 0.0  &  60 & 200 & 0.0436 & 0.0490 & 0.0496 & 0.0314 & 0.0381 & 0.0386 \\
0.4 & 0.0  &  80 & 200 & 0.0455 & 0.0539 & 0.0547 & 0.0329 & 0.0422 & 0.0428 \\
0.4 & 0.0  & 100 & 200 & 0.0478 & 0.0586 & 0.0596 & 0.0344 & 0.0461 & 0.0468 \\
0.5 & 0.0  &  40 & 200 & 0.0431 & 0.0431 & 0.0436 & 0.0314 & 0.0352 & 0.0356 \\
0.5 & 0.0  &  60 & 200 & 0.0460 & 0.0485 & 0.0491 & 0.0335 & 0.0398 & 0.0403 \\
0.5 & 0.0  &  80 & 200 & 0.0486 & 0.0535 & 0.0542 & 0.0353 & 0.0440 & 0.0447 \\
0.5 & 0.0  & 100 & 200 & 0.0503 & 0.0581 & 0.0590 & 0.0363 & 0.0480 & 0.0488 \\
\hline
0.3 & $-$0.3 &  40 & 200 & 0.0436 & 0.0439 & 0.0444 & 0.0311 & 0.0362 & 0.0366 \\
0.3 & $-$0.3 &  60 & 200 & 0.0470 & 0.0493 & 0.0501 & 0.0336 & 0.0409 & 0.0415 \\
0.3 & $-$0.3 &  80 & 200 & 0.0492 & 0.0543 & 0.0552 & 0.0352 & 0.0452 & 0.0460 \\
0.3 & $-$0.3 & 100 & 200 & 0.0509 & 0.0590 & 0.0601 & 0.0364 & 0.0493 & 0.0503 \\
0.4 & $-$0.3 &  40 & 200 & 0.0407 & 0.0435 & 0.0440 & 0.0294 & 0.0337 & 0.0340 \\
0.4 & $-$0.3 &  60 & 200 & 0.0434 & 0.0489 & 0.0496 & 0.0315 & 0.0381 & 0.0386 \\
0.4 & $-$0.3 &  80 & 200 & 0.0458 & 0.0539 & 0.0547 & 0.0330 & 0.0422 & 0.0428 \\
0.4 & $-$0.3 & 100 & 200 & 0.0475 & 0.0586 & 0.0596 & 0.0341 & 0.0460 & 0.0468 \\
0.5 & $-$0.3 &  40 & 200 & 0.0434 & 0.0431 & 0.0436 & 0.0315 & 0.0352 & 0.0356 \\
0.5 & $-$0.3 &  60 & 200 & 0.0464 & 0.0485 & 0.0491 & 0.0338 & 0.0398 & 0.0403 \\
0.5 & $-$0.3 &  80 & 200 & 0.0486 & 0.0534 & 0.0542 & 0.0352 & 0.0440 & 0.0447 \\
0.5 & $-$0.3 & 100 & 200 & 0.0504 & 0.0581 & 0.0590 & 0.0364 & 0.0480 & 0.0488 \\
\hline\hline
\end{tabular}
\begin{tablenotes}[flushleft]
\small
\item \textbf{Simulation Setting}: Same as Table~\ref{tab:bias_results}.
Nominal one-sided $\alpha = 0.025$.
\item \textbf{Columns}: Observed = empirical Type~I error from simulations;
Est = plugin estimate using Equation~\eqref{eq:type_one_error_z} or \eqref{eq:type_one_error_bin} with bias from Equation~\eqref{eq:bias_final};
Est$_{\text{max}}$ = conservative estimate using maximum bias bound (Equation~\eqref{eq:bias_max}).
\end{tablenotes}
\end{threeparttable}
\end{table}

\clearpage

\begin{table}[htbp]
\centering
\small
\begin{threeparttable}
\caption{Type I Error Rate for TTE Confirmatory Endpoints by Efficacy--Survival Correlation}
\label{tab:tte_type1}
\setlength\tabcolsep{2.5pt}
\begin{tabular}{cccc|cc|cc|cc|cc}
\hline\hline
\multicolumn{4}{c|}{\textbf{Design Parameters}} &
\multicolumn{2}{c|}{\textbf{Landmark Z}} &
\multicolumn{2}{c|}{\textbf{One-sample Exp}} &
\multicolumn{2}{c|}{\textbf{Two-sample LR}} &
\multicolumn{2}{c}{\textbf{Cox Score}} \\
\hline
$p$ & $\rho_c$ & $n_1$ & $n_1\!+\!n_2$ &
Obs & Est &
Obs & Est &
Obs & Est &
Obs & Est \\
\hline
0.3 & 0.7 &  40 & 200 & 0.0301 & 0.0311 & 0.0281 & 0.0335 & 0.0307 & 0.0308 & 0.0312 & 0.0308 \\
0.3 & 0.7 &  60 & 200 & 0.0309 & 0.0326 & 0.0293 & 0.0357 & 0.0314 & 0.0322 & 0.0320 & 0.0322 \\
0.3 & 0.7 &  80 & 200 & 0.0321 & 0.0338 & 0.0303 & 0.0376 & 0.0327 & 0.0334 & 0.0328 & 0.0334 \\
0.3 & 0.7 & 100 & 200 & 0.0329 & 0.0350 & 0.0314 & 0.0393 & 0.0339 & 0.0345 & 0.0342 & 0.0345 \\
0.4 & 0.7 &  40 & 200 & 0.0295 & 0.0303 & 0.0280 & 0.0332 & 0.0307 & 0.0306 & 0.0309 & 0.0306 \\
0.4 & 0.7 &  60 & 200 & 0.0308 & 0.0315 & 0.0293 & 0.0353 & 0.0316 & 0.0320 & 0.0319 & 0.0320 \\
0.4 & 0.7 &  80 & 200 & 0.0315 & 0.0327 & 0.0302 & 0.0372 & 0.0328 & 0.0332 & 0.0332 & 0.0332 \\
0.4 & 0.7 & 100 & 200 & 0.0322 & 0.0337 & 0.0313 & 0.0389 & 0.0335 & 0.0342 & 0.0344 & 0.0342 \\
0.5 & 0.7 &  40 & 200 & 0.0290 & 0.0295 & 0.0276 & 0.0327 & 0.0305 & 0.0302 & 0.0311 & 0.0302 \\
0.5 & 0.7 &  60 & 200 & 0.0299 & 0.0305 & 0.0288 & 0.0346 & 0.0314 & 0.0315 & 0.0317 & 0.0315 \\
0.5 & 0.7 &  80 & 200 & 0.0308 & 0.0315 & 0.0301 & 0.0364 & 0.0324 & 0.0327 & 0.0328 & 0.0327 \\
0.5 & 0.7 & 100 & 200 & 0.0311 & 0.0323 & 0.0307 & 0.0380 & 0.0335 & 0.0337 & 0.0339 & 0.0337 \\
\hline
0.3 & 0.3 &  40 & 200 & 0.0265 & 0.0273 & 0.0240 & 0.0282 & 0.0277 & 0.0272 & 0.0277 & 0.0272 \\
0.3 & 0.3 &  60 & 200 & 0.0269 & 0.0278 & 0.0246 & 0.0290 & 0.0277 & 0.0277 & 0.0281 & 0.0277 \\
0.3 & 0.3 &  80 & 200 & 0.0277 & 0.0282 & 0.0251 & 0.0296 & 0.0283 & 0.0282 & 0.0288 & 0.0282 \\
0.3 & 0.3 & 100 & 200 & 0.0277 & 0.0286 & 0.0254 & 0.0302 & 0.0289 & 0.0286 & 0.0292 & 0.0286 \\
0.4 & 0.3 &  40 & 200 & 0.0266 & 0.0272 & 0.0240 & 0.0282 & 0.0274 & 0.0272 & 0.0275 & 0.0272 \\
0.4 & 0.3 &  60 & 200 & 0.0269 & 0.0277 & 0.0249 & 0.0290 & 0.0280 & 0.0278 & 0.0283 & 0.0277 \\
0.4 & 0.3 &  80 & 200 & 0.0274 & 0.0281 & 0.0250 & 0.0296 & 0.0284 & 0.0282 & 0.0287 & 0.0282 \\
0.4 & 0.3 & 100 & 200 & 0.0274 & 0.0285 & 0.0254 & 0.0302 & 0.0287 & 0.0286 & 0.0293 & 0.0286 \\
0.5 & 0.3 &  40 & 200 & 0.0265 & 0.0271 & 0.0240 & 0.0281 & 0.0273 & 0.0272 & 0.0277 & 0.0272 \\
0.5 & 0.3 &  60 & 200 & 0.0270 & 0.0276 & 0.0245 & 0.0289 & 0.0279 & 0.0277 & 0.0283 & 0.0277 \\
0.5 & 0.3 &  80 & 200 & 0.0271 & 0.0280 & 0.0249 & 0.0295 & 0.0282 & 0.0281 & 0.0288 & 0.0281 \\
0.5 & 0.3 & 100 & 200 & 0.0276 & 0.0283 & 0.0256 & 0.0301 & 0.0286 & 0.0285 & 0.0290 & 0.0285 \\
\hline
0.3 & 0.0 &  40 & 200 & 0.0245 & 0.0251 & 0.0212 & 0.0251 & 0.0254 & 0.0250 & 0.0254 & 0.0250 \\
0.3 & 0.0 &  60 & 200 & 0.0243 & 0.0251 & 0.0212 & 0.0251 & 0.0253 & 0.0250 & 0.0256 & 0.0250 \\
0.3 & 0.0 &  80 & 200 & 0.0243 & 0.0251 & 0.0211 & 0.0251 & 0.0256 & 0.0250 & 0.0257 & 0.0250 \\
0.3 & 0.0 & 100 & 200 & 0.0247 & 0.0251 & 0.0214 & 0.0251 & 0.0254 & 0.0250 & 0.0257 & 0.0250 \\
0.4 & 0.0 &  40 & 200 & 0.0244 & 0.0251 & 0.0210 & 0.0251 & 0.0254 & 0.0250 & 0.0255 & 0.0250 \\
0.4 & 0.0 &  60 & 200 & 0.0245 & 0.0251 & 0.0212 & 0.0251 & 0.0253 & 0.0250 & 0.0256 & 0.0250 \\
0.4 & 0.0 &  80 & 200 & 0.0245 & 0.0251 & 0.0209 & 0.0251 & 0.0254 & 0.0250 & 0.0256 & 0.0250 \\
0.4 & 0.0 & 100 & 200 & 0.0245 & 0.0251 & 0.0213 & 0.0251 & 0.0254 & 0.0250 & 0.0255 & 0.0250 \\
0.5 & 0.0 &  40 & 200 & 0.0244 & 0.0251 & 0.0209 & 0.0251 & 0.0254 & 0.0250 & 0.0255 & 0.0250 \\
0.5 & 0.0 &  60 & 200 & 0.0243 & 0.0251 & 0.0210 & 0.0251 & 0.0253 & 0.0250 & 0.0255 & 0.0250 \\
0.5 & 0.0 &  80 & 200 & 0.0247 & 0.0251 & 0.0211 & 0.0251 & 0.0256 & 0.0250 & 0.0254 & 0.0250 \\
0.5 & 0.0 & 100 & 200 & 0.0245 & 0.0251 & 0.0209 & 0.0251 & 0.0254 & 0.0250 & 0.0255 & 0.0250 \\
\hline\hline
\end{tabular}
\begin{tablenotes}[flushleft]
\small
\item \textbf{Simulation Setting}: Null hypothesis $p_L = p_H = p$, $q_L = q_H = 0.8$, $\phi = 0$, exponential survival with $\lambda_0 = 0.1$, landmark at $\tau = 24$ weeks. Results averaged over $10^6$ replications. Nominal one-sided $\alpha = 0.025$. Results for $\phi = -0.3$ are virtually identical (differences $< 0.0004$), confirming negligible sensitivity to the efficacy--safety correlation.
\item \textbf{Columns}: Obs = empirical Type~I error from simulations; Est = plugin estimate. For the landmark Z-test, the plugin uses Equation~\eqref{eq:type1_landmark} with $\Cov(S(\tau), U)$ estimated from Stage~1 data. For the one-sample exponential and two-sample log-rank tests, the plugin applies the delta method chain through the log-hazard scale. The Cox score test uses the partial likelihood score statistic; the Wald test plugin (Section~\ref{subsubsec:cox_test}) serves as its approximation.
\item \textbf{$\rho_c$}: Gaussian copula parameter controlling the efficacy--survival correlation. Estimated Pearson correlations $\hat{\rho}_{TX} \approx 0.53$ ($\rho_c = 0.7$), $0.22$ ($\rho_c = 0.3$), and $0.00$ ($\rho_c = 0$).
\end{tablenotes}
\end{threeparttable}
\end{table}

\clearpage

\begin{table}[htbp]
\centering
\small
\begin{threeparttable}
\caption{Summary: Comparison of Z-test and Binomial Test Performance}
\label{tab:test_summary}
\begin{tabular}{lcc}
\hline\hline
\textbf{Metric} & \textbf{Z-test} & \textbf{Binomial Test} \\
\hline
Mean observed Type I error & 0.0464 & 0.0334 \\
Mean plugin estimate & 0.0512 & 0.0416 \\
Mean overestimation (Est $-$ Obs) & $+$0.0049 & $+$0.0081 \\
Inflation factor (vs $\alpha = 0.025$) & 1.86$\times$ & 1.34$\times$ \\
Plugin conservative rate & 96\% & 100\% \\
Estimator property & Conservative & Conservative \\
\hline\hline
\end{tabular}
\begin{tablenotes}[flushleft]
\small
\item \textbf{Note}: Summary statistics computed across 24 null hypothesis scenarios ($n_1 \in \{40, 60, 80, 100\}$).
Both plugin estimators are conservative (overestimate the true error rate),
providing reliable upper bounds for regulatory planning.
The binomial test has lower inflation ($1.34\times$ vs $1.86\times$)
and a more conservative plugin (100\% vs 96\%).
\end{tablenotes}
\end{threeparttable}
\end{table}

\begin{table}[htbp]
\centering
\small
\begin{threeparttable}
\caption{Summary: TTE Type I Error Inflation by Efficacy--Survival Correlation}
\label{tab:tte_summary}
\begin{tabular}{lcccccc}
\hline\hline
& \multicolumn{2}{c}{$\rho_c = 0.7$} & \multicolumn{2}{c}{$\rho_c = 0.3$} & \multicolumn{2}{c}{$\rho_c = 0$} \\
\cmidrule(lr){2-3} \cmidrule(lr){4-5} \cmidrule(lr){6-7}
\textbf{Test} & Mean & Factor & Mean & Factor & Mean & Factor \\
\hline
Binary Z-test (reference)  & 0.0463 & 1.85$\times$ & 0.0463 & 1.85$\times$ & 0.0464 & 1.85$\times$ \\
Binary binomial (reference)& 0.0334 & 1.34$\times$ & 0.0334 & 1.34$\times$ & 0.0335 & 1.34$\times$ \\
\hline
Landmark Z-test            & 0.0309 & 1.24$\times$ & 0.0271 & 1.08$\times$ & 0.0245 & 0.98$\times$ \\
One-sample exponential        & 0.0296 & 1.18$\times$ & 0.0248 & 0.99$\times$ & 0.0211 & 0.84$\times$ \\
Two-sample log-rank           & 0.0321 & 1.28$\times$ & 0.0281 & 1.12$\times$ & 0.0254 & 1.02$\times$ \\
Two-sample Cox score          & 0.0325 & 1.30$\times$ & 0.0284 & 1.14$\times$ & 0.0256 & 1.02$\times$ \\
\hline
$\hat{\rho}_{TX}$          & \multicolumn{2}{c}{0.53} & \multicolumn{2}{c}{0.22} & \multicolumn{2}{c}{0.00} \\
\hline\hline
\end{tabular}
\begin{tablenotes}[flushleft]
\small
\item \textbf{Note}: Mean = mean observed Type~I error across 12 scenarios ($\phi = 0$) at given $\rho_c$. Factor = ratio to nominal $\alpha = 0.025$. Binary test results are included as reference and are invariant to $\rho_c$ (small differences reflect Monte Carlo variability). The one-sample exponential test at $\rho_c = 0$ shows rates below nominal due to the log-scale test's inherent conservatism, not selection bias. The Cox score test shows slightly higher inflation than the log-rank test; both share the same plugin estimator. $\hat{\rho}_{TX}$: estimated Pearson correlation between binary efficacy and survival time.
\end{tablenotes}
\end{threeparttable}
\end{table}

\clearpage
\appendix
\section{Mathematical Proofs and Technical Derivations}
\label{app:proofs}

\paragraph{Notation.}
Throughout this appendix, $n_1$ denotes the per-arm sample size in the dose optimization stage (Stage~1) and $n_2$ the number of additional patients enrolled on the selected dose in the confirmatory stage (Stage~2). Under the null hypothesis, $p_L = p_H = p$, $q_L = q_H = q$, the utility scores for both doses share common moments $\mu = \mathbb{E}[U]$ and $\sigma_U^2 = \Var(U)$. The symbol $\phi(\cdot)$ denotes the standard normal density $(2\pi)^{-1/2}\exp(-x^2/2)$, $\Phi(\cdot)$ denotes the standard normal CDF, and $\pi \approx 3.14159$ denotes the mathematical constant (distinguished from the probability vector $\boldsymbol{\pi}$ by context).

\subsection{Truncated Normal Expectation}
\label{app:truncated}

\begin{lemma}[Truncated Selection Expectation]
\label{lem:truncated}
For $Z_H, Z_L \stackrel{\mathrm{iid}}{\sim} N(0,1)$ and any threshold $k \in \mathbb{R}$:
\begin{equation}
\mathbb{E}\bigl[Z_H \cdot \mathbf{1}\{Z_H - Z_L > k\}\bigr]
+ \mathbb{E}\bigl[Z_L \cdot \mathbf{1}\{Z_H - Z_L \leq k\}\bigr]
= \frac{1}{\sqrt{\pi}}\, e^{-k^2/4}.
\label{eq:lem_truncated}
\end{equation}
\end{lemma}

\begin{proof}
Define the orthogonal rotation $W = (Z_H + Z_L)/\sqrt{2}$, $V = (Z_H - Z_L)/\sqrt{2}$, so that $W, V \stackrel{\mathrm{iid}}{\sim} N(0,1)$. The inverse relations are $Z_H = (W+V)/\sqrt{2}$ and $Z_L = (W-V)/\sqrt{2}$, and the selection event becomes $\{Z_H - Z_L > k\} = \{V > k/\sqrt{2}\}$.

\medskip\noindent\textit{High-dose term.}
By independence of $W$ and $V$:
\begin{align*}
\mathbb{E}\bigl[Z_H \cdot \mathbf{1}\{V > k/\sqrt{2}\}\bigr]
&= \frac{1}{\sqrt{2}}\,
   \Bigl[\underbrace{\mathbb{E}[W]}_{=\,0}
         \cdot \Pr(V > k/\sqrt{2})
         + \mathbb{E}\bigl[V \cdot \mathbf{1}\{V > k/\sqrt{2}\}\bigr]
   \Bigr] \\
&= \frac{1}{\sqrt{2}}\,
   \mathbb{E}\bigl[V \cdot \mathbf{1}\{V > k/\sqrt{2}\}\bigr].
\end{align*}

\noindent\textit{Low-dose term.}
An identical calculation gives:
\[
\mathbb{E}\bigl[Z_L \cdot \mathbf{1}\{V \leq k/\sqrt{2}\}\bigr]
= -\frac{1}{\sqrt{2}}\,
   \mathbb{E}\bigl[V \cdot \mathbf{1}\{V \leq k/\sqrt{2}\}\bigr].
\]

\noindent\textit{Standard truncated-normal identities.}
For $V \sim N(0,1)$ and any $a \in \mathbb{R}$, direct integration yields
\begin{equation}
\mathbb{E}\bigl[V \cdot \mathbf{1}\{V > a\}\bigr]
= \int_a^\infty v\,\phi(v)\,dv
= \bigl[-\phi(v)\bigr]_a^\infty
= \phi(a),
\label{eq:trunc_identity_upper}
\end{equation}
and, by complementarity, $\mathbb{E}\bigl[V \cdot \mathbf{1}\{V \leq a\}\bigr] = -\phi(a)$.

\medskip\noindent\textit{Combining.}
Setting $a = k/\sqrt{2}$:
\begin{align*}
\text{LHS of \eqref{eq:lem_truncated}}
&= \frac{1}{\sqrt{2}}\bigl[\phi(k/\sqrt{2})
   - \bigl(-\phi(k/\sqrt{2})\bigr)\bigr]
 = \frac{2\,\phi(k/\sqrt{2})}{\sqrt{2}}
 = \sqrt{2}\,\phi\!\left(\frac{k}{\sqrt{2}}\right) \\
&= \sqrt{2} \cdot \frac{1}{\sqrt{2\pi}}
   \exp\!\left(-\frac{k^2}{4}\right)
 = \frac{1}{\sqrt{\pi}}\,e^{-k^2/4}. \qedhere
\end{align*}
\end{proof}

\subsection{Selection Bias for Binary Endpoints}
\label{app:bias_binary}

\begin{theorem}[Selection Bias Under Utility-Based Selection]
\label{thm:bias_general}
Let $X$ be a binary efficacy endpoint with $\mathbb{E}[X] = p$, and let $U = \sum_{k=1}^4 u_k W_k$ be the utility score with $\mathbb{E}[U] = \mu$ and $\Var(U) = \sigma_U^2$. Under the null hypothesis ($p_L = p_H = p$, $q_L = q_H = q$) and the selection rule ``select Dose~$H$ if $\bar{U}_H - \bar{U}_L > \lambda_u$; otherwise select Dose~$L$,'' the selection-induced bias in the response rate of the selected dose is:
\begin{equation}
\Bias(\hat{p}_{\mathrm{selected}})
= \frac{\Cov(X, U)}{\sigma_U \sqrt{n_1}}
  \cdot \frac{1}{\sqrt{\pi}}\,
  \exp\!\left(-\frac{\lambda_u^2 n_1}{4\sigma_U^2}\right),
\label{eq:bias_general_app}
\end{equation}
where $\Cov(X,U) = \sum_{k=1}^{4} x_k u_k \pi_k - p\mu$ with $x_k \in \{0,1\}$ denoting the efficacy indicator for outcome $k$. More generally, for any endpoint $W$ observed jointly with $U$ and satisfying CLT-based joint normality of $(\bar{W}_j, \bar{U}_j)$, the formula holds with $\Cov(X,U)$ replaced by $\Cov(W,U)$.
\end{theorem}

\begin{proof}
The derivation is given in Section~\ref{subsec:bias_derivation} of the main text; we summarize the key steps for reference.

Under the null, the standardized utility means $Z_j = (\bar{U}_j - \mu)/(\sigma_U/\sqrt{n_1})$ are i.i.d.\ $N(0,1)$ for $j \in \{H,L\}$. By the CLT, $(\hat{p}_j, \bar{U}_j)$ are approximately jointly normal, giving the conditional expectation
\[
\mathbb{E}[\hat{p}_j - p \mid \bar{U}_j]
= \frac{\Cov(X, U)}{\sigma_U^2}\,(\bar{U}_j - \mu).
\]
The law of iterated expectations yields
\[
\Bias(\hat{p}_{\mathrm{selected}})
= \frac{\Cov(X, U)}{\sigma_U \sqrt{n_1}}\,
  \mathbb{E}\bigl[
    Z_H \cdot \mathbf{1}\{Z_H - Z_L > k\}
    + Z_L \cdot \mathbf{1}\{Z_H - Z_L \leq k\}
  \bigr],
\]
where $k = \lambda_u \sqrt{n_1}/\sigma_U$. Lemma~\ref{lem:truncated} evaluates the expectation as $(1/\sqrt{\pi})\exp(-k^2/4)$, yielding \eqref{eq:bias_general_app}.

The generalization to arbitrary $W$ follows identically: replace $X$ by $W$ throughout, noting that only the covariance $\Cov(W, U)$ and the CLT approximation for $(\bar{W}_j, \bar{U}_j)$ are needed.
\end{proof}

\begin{corollary}[Maximum Selection Bias]
\label{cor:bias_max}
When selection is based solely on efficacy ($u_1 = u_2 = 1$, $u_3 = u_4 = 0$, i.e., $U = X$) with no threshold ($\lambda_u = 0$):
\begin{equation}
\Bias_{\max}(\hat{p})
= \frac{\sqrt{p(1-p)}}{\sqrt{n_1 \pi}},
\label{eq:bias_max_app}
\end{equation}
recovering the result of \citet{bauer2010selection} with $m_1(2) = 1/\sqrt{\pi}$.
\end{corollary}

\begin{proof}
When $U = X$: $\Cov(X,U) = \Var(X) = p(1-p)$ and $\sigma_U = \sqrt{p(1-p)}$, so $\Cov(X,U)/\sigma_U = \sqrt{p(1-p)}$. Setting $\lambda_u = 0$ in \eqref{eq:bias_general_app} gives the result.
\end{proof}

\subsection{Bias Propagation via the Delta Method}
\label{app:bias_tte_derivations}

The following proposition establishes all TTE bias formulas used in the main text. Throughout, $\mathcal{B}$ denotes the common factor:
\begin{equation}
\mathcal{B}
\;\equiv\; \Bias(\bar{T})
= \frac{\Cov(T, U)}{\sigma_U \sqrt{n_1}}
  \cdot \frac{1}{\sqrt{\pi}}\,
  \exp\!\left(-\frac{\lambda_u^2 n_1}{4\sigma_U^2}\right),
\label{eq:calB_def}
\end{equation}
which follows from Theorem~\ref{thm:bias_general} with $W = T$.

\begin{proposition}[Bias Propagation Chain]
\label{prop:bias_chain}
Under exponential survival with true hazard $\lambda_0$ (mean survial time $1/\lambda_0$) and the dose-selection mechanism of Theorem~\ref{thm:bias_general}, let $w_1 = n_1/(n_1+n_2)$ denote the Stage~1 dilution factor:
\begin{enumerate}
\item[\textup{(a)}]
\textbf{Landmark survival rate.}
For the $\tau$-month indicator $S_i(\tau) = \mathbf{1}\{T_i > \tau\}$:
\begin{equation}
\Bias\!\bigl(\hat{S}(\tau)\bigr)_{\mathrm{comb}}
= w_1 \cdot \frac{\Cov(S(\tau),\, U)}{\sigma_U \sqrt{n_1}}
  \cdot \frac{1}{\sqrt{\pi}}\,
  \exp\!\left(-\frac{\lambda_u^2 n_1}{4\sigma_U^2}\right).
\label{eq:bias_landmark_app}
\end{equation}

\item[\textup{(b)}]
\textbf{Hazard rate} (first delta method step).
With $\hat{\lambda} = 1/\bar{T}$ and $g(x) = 1/x$, $g'(1/\lambda_0) = -\lambda_0^2$:
\begin{equation}
\Bias(\hat{\lambda})_{\mathrm{comb}}
\;\approx\; -\lambda_0^2 \cdot w_1 \cdot \mathcal{B}.
\label{eq:bias_lambda_chain}
\end{equation}

\item[\textup{(c)}]
\textbf{Log hazard ratio} (second delta method step).
For the Cox coefficient $\hat{\beta} = \log(\hat{\lambda}_{\mathrm{trt}}/\hat{\lambda}_{\mathrm{ctrl}})$, where only the treatment arm is affected by selection:
\begin{equation}
\Bias(\hat{\beta})_{\mathrm{comb}}
\;\approx\; -\lambda_0 \cdot w_1 \cdot \mathcal{B}.
\label{eq:bias_beta_chain}
\end{equation}
\end{enumerate}
\end{proposition}

\begin{proof}
\textit{Part~(a).}
$S_i(\tau) = \mathbf{1}\{T_i > \tau\}$ is binary with the same algebraic structure as the response indicator $X$. Since Stage~2 patients are enrolled after dose selection and their outcomes are therefore independent of the selection mechanism, $\mathbb{E}[\hat{S}_{\text{stage2}}(\tau)] = S_0(\tau)$ under $H_0$, contributing zero additional bias. Equation~\eqref{eq:bias_landmark_app} follows from the combined-stage weighted average with dilution factor $w_1$.

\medskip
\textit{Part~(b).}
Under $H_0$, $\mathbb{E}[\bar{T}] = 1/\lambda_0$. The first-order Taylor expansion of $g(x) = 1/x$ around $x_0 = 1/\lambda_0$ gives:
\[
\hat{\lambda} = g(\bar{T})
\approx \lambda_0 - \lambda_0^2\bigl(\bar{T} - 1/\lambda_0\bigr).
\]
Taking expectations: $\Bias(\hat{\lambda}) \approx -\lambda_0^2 \cdot \mathcal{B}$. Since $\mathcal{B} > 0$ when $\Cov(T,U) > 0$, the hazard is biased downward---the selected dose appears to have better survival than reality. Applying the dilution factor $w_1$ yields \eqref{eq:bias_lambda_chain}.

\medskip
\textit{Part~(c).}
\textit{Log transformation.}
The Cox coefficient is $\hat{\beta} = \log\hat{\lambda}_{\mathrm{trt}} - \log\hat{\lambda}_{\mathrm{ctrl}}$. Since the control arm is unaffected by selection, $\mathbb{E}[\log\hat{\lambda}_{\mathrm{ctrl}}] = \log\lambda_0$ to first order. Applying $h(\lambda) = \log\lambda$ with $h'(\lambda_0) = 1/\lambda_0$:
\[
\Bias(\hat{\beta})_{\mathrm{comb}}
\approx \frac{\Bias(\hat{\lambda})_{\mathrm{comb}}}{\lambda_0}
= \frac{-\lambda_0^2 \cdot w_1 \cdot \mathcal{B}}{\lambda_0}
= -\lambda_0 \cdot w_1 \cdot \mathcal{B}.
\]

\end{proof}

\subsection{Type~I Error for TTE Confirmatory Tests}
\label{app:tte_type1_proofs}

All three TTE confirmatory tests share a common structure: (i)~the bias from Proposition~\ref{app:bias_tte_derivations} shifts the test statistic under $H_0$; (ii)~the shifted-normal model yields the inflated Type~I error. We state each result as a part of a single proposition and then give a unified proof.

\begin{proposition}[Type~I Error for TTE Confirmatory Tests]
\label{prop:tte_type1_unified}
Under the two-stage design with pooled data ($n_1 + n_2$ patients on the selected dose), the Type~I error for each confirmatory test at nominal one-sided level~$\alpha$ is:
\begin{enumerate}
\item[\textup{(a)}]
\textbf{Landmark survival Z-test}
($H_0\!: S(\tau) \leq S_0$; reject when $Z \geq z_{1-\alpha}$):
\begin{equation}
\operatorname{Type~I~Error}_{\mathrm{Landmark}} = 1 - \Phi\!\left(z_{1-\alpha} - \frac{\Bias(\hat{S}(\tau))_{\mathrm{comb}}} {\sqrt{S_0(1-S_0)/(n_1+n_2)}}\right),
\label{eq:type1_landmark_app}
\end{equation}
where $\Bias(\hat{S}(\tau))_{\mathrm{comb}}$ is given by Equation~\eqref{eq:bias_landmark_app}.

\item[\textup{(b)}]
\textbf{One-sample exponential test} ($H_0\!: \lambda \geq \lambda_0$; reject when $Z \leq -z_{1-\alpha}$):
\begin{equation}
\operatorname{Type~I~Error}_{\mathrm{Exp}} = \Phi\!\left(-z_{1-\alpha} - \Bias(\hat{\lambda})_{\mathrm{comb}} \cdot \sqrt{D}\right),
\label{eq:type1_exp_app}
\end{equation}
where $D$ is the expected number of events under $H_0$.

\item[\textup{(c)}]
\textbf{Two-sample Cox Wald test}
($H_0\!: \beta \geq 0$; reject when $Z \leq -z_{1-\alpha}$):
\begin{equation}
\operatorname{Type~I~Error}_{\mathrm{Cox}} = \Phi\!\left(-z_{1-\alpha} - \Bias(\hat{\beta})_{\mathrm{comb}} \cdot \sqrt{\frac{D_{\mathrm{total}}}{4}}\right),
\label{eq:type1_cox_app}
\end{equation}
where $D_{\mathrm{total}}$ is the total events across both arms and $\mathcal{I} = D_{\mathrm{total}}/4$ is the Fisher information under equal allocation.
\end{enumerate}

In all three cases, $\Cov(T,U) > 0$ implies that the Type~I error exceeds the nominal level~$\alpha$. When $\Cov(T,U) = 0$, all bias terms vanish and the formulas return the nominal~$\alpha$.
\end{proposition}

\begin{proof}
Each case follows three steps: identify the bias in the test statistic, specify the rejection region, and apply the shifted normal model.

\medskip
\noindent\textbf{(a) Landmark survival Z-test.}
The indicator $S_i(\tau) = \mathbf{1}\{T_i > \tau\}$ is binary, so the bias analysis is structurally identical to the binary response rate case in Section~\ref{sec:type_one_error_binary} of the main text. The combined-stage bias is $\Bias(\hat{S}(\tau))_{\mathrm{comb}} = w_1 \cdot \Bias(\hat{S}(\tau))$ by the dilution principle (Equation~\eqref{eq:combined_bias}). The test statistic
\[
Z_{\mathrm{Landmark}} = \frac{\hat{S}(\tau)_{\mathrm{comb}} - S_0}{\SE_0}, \qquad \SE_0 = \sqrt{\frac{S_0(1-S_0)}{n_1+n_2}},
\]
has bias $\Bias(Z_{\mathrm{Landmark}}) = \Bias(\hat{S}(\tau))_{\mathrm{comb}} / \SE_0 > 0$ under positive efficacy--survival correlation. Since rejection occurs in the upper tail ($Z \geq z_{1-\alpha}$) and
$Z_{\mathrm{Landmark}} \sim N(\Bias(Z_{\mathrm{Landmark}}), 1)$ approximately:
\[
\Pr(Z_{\mathrm{Landmark}} \geq z_{1-\alpha}) = 1 - \Phi\bigl(z_{1-\alpha} - \Bias(Z_{\mathrm{Landmark}})\bigr),
\]
yielding \eqref{eq:type1_landmark_app}.

\medskip
\noindent\textbf{(b) One-sample exponential test.}
By Proposition~\ref{prop:bias_chain}(b)--(c), the combined-stage log-hazard bias is:
\[
\Bias(\log\hat{\lambda})_{\mathrm{comb}} = \frac{\Bias(\hat{\lambda})_{\mathrm{comb}}}{\lambda_0} = -\lambda_0 \cdot w_1 \cdot \mathcal{B}.
\]

The test statistic is
$Z_{\mathrm{Exp}} = (\log\hat{\lambda} - \log\lambda_0)\sqrt{D}$, where $\Var(\log\hat{\lambda}) \approx 1/D$ by the exponential MLE theory~\citep{lawless2003statistical}. The bias in the test statistic is:
\[
\Bias(Z_{\mathrm{Exp}}) = \Bias(\log\hat{\lambda})_{\mathrm{comb}} \cdot \sqrt{D} = -\lambda_0 \cdot w_1 \cdot \mathcal{B} \cdot \sqrt{D} < 0.
\]

Rejection occurs in the lower tail ($Z_{\mathrm{Exp}} \leq -z_{1-\alpha}$). Under $Z_{\mathrm{Exp}} \sim N(\Bias(Z_{\mathrm{Exp}}),\, 1)$:
\begin{align*}
\Pr(Z_{\mathrm{Exp}} \leq -z_{1-\alpha})
&= \Phi\bigl(-z_{1-\alpha} - \Bias(Z_{\mathrm{Exp}})\bigr) \\
&= \Phi\bigl(-z_{1-\alpha}
   + \lambda_0 \cdot w_1 \cdot \mathcal{B} \cdot \sqrt{D}\bigr),
\end{align*}
yielding \eqref{eq:type1_exp_app}.

\medskip
\noindent\textbf{(c) Two-sample Cox Wald test.}
By Proposition~\ref{prop:bias_chain}(c), $\Bias(\hat{\beta})_{\mathrm{comb}} = -\lambda_0 \cdot w_1 \cdot \mathcal{B}$. The Cox Wald statistic is:
\[
Z_{\mathrm{Cox}} = \hat{\beta} \cdot \sqrt{\frac{D_{\mathrm{total}}}{4}}.
\]
Under $H_0$ without selection bias, $\mathbb{E}[\hat{\beta}] = 0$. With selection bias, $\mathbb{E}[\hat{\beta}] = \Bias(\hat{\beta})_{\mathrm{comb}}$. By linearity of expectation (no further approximation beyond the asymptotic normality of $\hat{\beta}$;
cf.~\citealp{barndorffnielsen1994inference}):
\[
\mathbb{E}[Z_{\mathrm{Cox}}] = \Bias(\hat{\beta})_{\mathrm{comb}} \cdot \sqrt{\frac{D_{\mathrm{total}}}{4}}.
\]

Under $Z_{\mathrm{Cox}} \sim N\!\left(\mathbb{E}[Z_{\mathrm{Cox}}],\, 1\right)$ with lower-tail rejection:
\begin{align*}
\Pr(Z_{\mathrm{Cox}} \leq -z_{1-\alpha})
&= \Phi\!\left(-z_{1-\alpha}
   - \Bias(\hat{\beta})_{\mathrm{comb}}
     \cdot \sqrt{\frac{D_{\mathrm{total}}}{4}}\right) \\
&= \Phi\!\left(-z_{1-\alpha}
   + \lambda_0 \cdot w_1 \cdot \mathcal{B}
   \cdot \sqrt{\frac{D_{\mathrm{total}}}{4}}\right),
\end{align*}
yielding \eqref{eq:type1_cox_app}.

\medskip
\noindent\textbf{Sign verification.}
In all three cases, $\Cov(T,U) > 0$ implies $\mathcal{B} > 0$ (Equation~\eqref{eq:calB_def}), which shifts the test statistic toward the rejection region:
\begin{itemize}
\item[(a)] $\Bias(Z_{\mathrm{Landmark}}) > 0$ shifts toward the upper-tail critical value;
\item[(b)] $\Bias(Z_{\mathrm{Exp}}) < 0$ shifts toward the lower-tail critical value;
\item[(c)] $\Bias(Z_{\mathrm{Cox}}) < 0$ shifts toward the lower-tail critical value.
\end{itemize}
In each case, the resulting Type~I error exceeds the nominal~$\alpha$.
\end{proof}

\end{document}